\documentclass[11pt]{article}
\usepackage[margin=1in]{geometry}
\usepackage[utf8]{inputenc}
\usepackage[T1]{fontenc}
\usepackage{lmodern}
\usepackage{amsmath,amssymb,amsthm,bm}
\usepackage{enumitem}
\usepackage{booktabs}
\usepackage{graphicx}
\usepackage{url}
\usepackage[round]{natbib}
\usepackage[hidelinks]{hyperref}
\newcommand{\Hline}{\hline}
\newcommand{\sref}[1]{\ref{#1}}
\newcommand{\logit}{\mathrm{logit}}
\newtheorem{theorem}{Theorem}
\newtheorem{lemma}{Lemma}
\newtheorem{proposition}{Proposition}
\newtheorem{corollary}{Corollary}
\theoremstyle{definition}
\newtheorem{definition}{Definition}
\newtheorem{assumption}{Assumption}
\theoremstyle{remark}
\newtheorem{remark}{Remark}
\providecommand{\E}{\mathrm{E}}
\newcommand{\expit}{\mathrm{expit}}
\newcommand{\ip}[2]{\langle #1,#2\rangle_{w}}
\newcommand{\nrm}[1]{\lVert #1\rVert_{w}}
\newcommand{\tx}{\tilde{x}}
\newcommand{\cG}{\mathcal{G}}
\newcommand{\cX}{\mathcal{X}}
\newcommand{\cH}{\mathcal{H}}
\newcommand{\cK}{\mathcal{K}}
\newcommand{\dto}{\rightsquigarrow}

\title{One Goodness-of-Fit Test Is Not Enough: Error-Controlled Localization of Misfit in Logistic Risk Models}
\author{Ebrahim Khaled Ebrahim\thanks{Corresponding author: ebrahimkhaled@alexu.edu.eg} \and Ahmed El-Kotory \and
Osama Abd El-Aziz Hussein\\[4pt] \small Department of Statistics, Faculty of Business, Alexandria University, Alexandria, Egypt}
\date{}

\begin{document}
\maketitle

\begin{abstract}
A goodness-of-fit test tells whether a logistic risk model fits, but the analyst also needs to know what is wrong:
an overall risk that is too high is repaired by updating the intercept, risks that are too extreme by recalibration,
a curved relation between score and risk by a flexible recalibration, and misfit within groups of patients only by
revising the model. We split any misfit into these four parts, which are orthogonal and match the calibration
hierarchy. Each part gets a group of score tests that sees that part only, and the failing parts are named by closed
testing. In external validation, in large samples and for small departures, a part is wrongly named with
probability at most the chosen level, and the highest
part named points to the lightest update that removes the misfit. Checking a model on its development data never sees
an error in the overall risk or the calibration slope, and with binary or skewed covariates it can mistake misfit
within patient groups for a curved link. External validation sees all four parts, for any prediction model. In
simulations it kept its error rate when the model was right and for small departures, and usually pointed
to the right update. On intensive-care data its verdicts for one model differed in-sample, on similar patients
and after transport.

\end{abstract}

\section{Introduction}
\label{s:intro}

A logistic risk model is checked twice in its life. It is first checked on the patients it was fitted to, usually
with the test of \citet{hosmerlemeshow1980} or one of its many successors \citep{hosmer1997, liu2024, ebrahim2026benchmark}. It is then
validated on new patients, for example in another hospital or a later period, with the coefficients frozen
\citep{steyerberg2010assessing, riley2021}. In both settings a goodness-of-fit test answers one question: does the
model fit? The analyst needs a second answer: if it does not, what is wrong? The remedy depends on it. If the
model predicts too many deaths overall, updating the intercept repairs it. If its risks are too extreme, a logistic
recalibration repairs it. If the relation between the model's score and the observed risk is curved, a flexible
recalibration or another link is needed. If the misfit lies within groups of patients that share a score, no
recalibration can help and the model must be revised \citep{steyerberg2004, janssen2008}.

A single test cannot give that second answer, and different tests answer differently. Every goodness-of-fit test
has power in a small set of directions only \citep{janssen2000global}, and the directions differ between the schools
of tests: partition tests that group patients by fitted risk \citep{hosmerlemeshow1980, tsiatis1980, pigeonheyse1999,
xie2008}, smoothed residual tests \citep{lecessie1991}, global residual sums \citep{osius1992}, directed score tests
\citep{stukel1988, pregibon1980goodness}, projection and empirical-process tests \citep{stute2002, escanciano2006,
liu2024}, and learned tests \citep{zhang2023}. A rejection by one of them says that the model fails somewhere in
that test's directions; a non-rejection says nothing about the others. Reading several tests side by side does not
solve the problem either, because the tests overlap and, without an error-control rule, the analyst who looks at many
tests will often name a problem that is not there.

This paper turns the second question into a multiple-testing problem with a definite answer. Section~\ref{s:parts}
splits any misfit into four orthogonal parts that match the levels of the calibration hierarchy of
\citet{vancalster2016}: calibration in the large, the calibration slope, the shape of the link between score and risk,
and the structure of the covariates. It also shows which parts can be seen at all. In-sample checking cannot see the
first two, ever, because the maximum-likelihood fit removes them; and with binary or skewed covariates it can mistake
covariate structure for a curved link. External validation sees all four. Section~\ref{s:method} builds a group of
tests for each part, combines each group by the Cauchy combination \citep{liuxie2020}, and names the parts that fail
by closed testing \citep{marcus1976closed}. The procedure controls the familywise error rate in the strong sense, it
finds every part its tests can see as the sample grows, and its highest named part points to the lightest update that
removes the misfit. Section~\ref{s:sim} reports simulations, including comparisons with single tests and with
simpler ways of naming the parts. Section~\ref{s:app} applies the procedure to the SUPPORT intensive-care data
\citep{knaus1995support}, where it gives three different verdicts for the spline model in three settings and where
testing the transported models again after the update their verdicts prescribe reveals a misfit within patient groups
that the first check missed, and to a multi-centre study of burn injuries, where the revision that repairs a model
in-sample does not make it transportable across centers.

Closed testing has been used to locate effects in genomics and functional data \citep{globaltestclosed2022,
domainfwer2022}, and calibration has been assessed by tests and confidence bands \citep{nattino2020, honestcal2022}.
To our knowledge no procedure locates the misfit of a logistic risk model with error control, or tells the analyst
which update the data call for. The theory draws on classical results, which we cite where they are used; the proofs
are in Appendix~A.

\section{Four Parts of a Misfit}
\label{s:parts}

\subsection{The decomposition}

Let $\pi(x)=\Pr(Y=1\mid X=x)$ and let $p_0(x)$ be the risks of the model being checked, with score
$\eta_0=\logit\,p_0$. Misfit is measured on the probability scale. A score test built on a function $b(x)$ reads the
misfit only through
\begin{equation}
\label{eq:read}
\E\{b(X)(Y-p_0(X))\}=\E\{b(X)(\pi(X)-p_0(X))\}=\langle b,g\rangle,\qquad g=\frac{\pi-p_0}{w_0},
\end{equation}
where $w_0=p_0(1-p_0)$ and $\langle f,g\rangle=\E\{w_0fg\}$ is the inner product weighted by the binomial variance.
We call $g$ the misfit: the difference between the true and the predicted risk, divided by its binomial
variance. Orthogonal below means uncorrelated after weighting by $p_0(1-p_0)$. With respect to this inner product, every misfit splits into four orthogonal parts:
\begin{itemize}
\item[] \textsc{intercept}: a constant, the overall risk being too high or too low;
\item[] \textsc{slope}: a multiple of the centered score, risks too extreme or too modest;
\item[] \textsc{link}: a function of the score orthogonal to the first two, a curved relation between score and risk;
\item[] \textsc{cov}: any function orthogonal to all functions of the score, misfit within groups of patients who
  share a score.
\end{itemize}
These are the levels of the calibration hierarchy of \citet{vancalster2016}: calibration in the large, weak, moderate
and strong calibration. They are also the rungs of the model-updating ladder \citep{steyerberg2004, janssen2008}:
the first part is repaired by updating the intercept, the first two by a logistic recalibration $a+b\eta_0$, the first
three by a flexible recalibration $f(\eta_0)$, and the fourth only by revising the model, since no function of the score
can repair misfit among patients who share a score (Appendix~A, Proposition~\ref{prop:update}). Locally, updating at the
highest rung that has a non-zero part removes the systematic misfit, up to estimation noise (for the link part, as
the flexibility of the recalibration grows), and updating at a lower rung leaves at least that part.

\subsection{What in-sample checking can see}

The maximum-likelihood fit of a logistic model with an intercept satisfies $\sum_i(Y_i-\hat p_i)=0$ and
$\sum_i(Y_i-\hat p_i)\hat\eta_i=0$, so on the development data the calibration intercept is $0$ and the slope $1$
exactly, whatever the truth. This identity is well known. Its consequence for testing is stronger: by the classical
argument of \citet{pierce1982}, refitting removes from every score statistic the part of the misfit that lies along
the model's own terms. In-sample checking therefore cannot see the \textsc{intercept} and \textsc{slope} parts at all, and
a test of them on the development data has power equal to its size (Theorem~\ref{thm:insample}). Only two questions are left
in-sample: is a function of the score missing (a link test, \citealp{pregibon1980goodness, stukel1988}), and is
there misfit that no function of the score explains? For a non-canonical link such as the probit the identities
fail, but the refit then removes that model's own intercept and slope directions instead.

Two further consequences matter in practice. First, a test that groups patients by fitted risk, as the
Hosmer--Lemeshow test does, sees only functions of the score; against misfit within groups of patients that is
orthogonal to every function of the score its local power equals its size, exactly in external validation (Corollary~\ref{cor:hl}). This is the reason, long seen in
simulations and real data \citep{hosmer1997, liu2024, ebrahim2026benchmark}, why these tests miss omitted interactions. Second, the in-sample link test
coincides with the external \textsc{link} part only when the average of each covariate given the score is a linear
function of the score, as it is for Gaussian main effects (Proposition~\ref{prop:alias}). With binary covariates, interactions or
skewed covariates the refit turns part of a covariate-structure misfit into a function of the score, and the
in-sample link test sees it. The in-sample \textsc{cov} test is never affected. We give a de-aliased link group that
reads the link part only, at the price of the link directions that the covariates themselves carry
(Proposition~\ref{prop:dealias}).

\section{The Localization Procedure}
\label{s:method}

\subsection{Groups of tests}

Each part gets a group of score tests whose bases lie in that part:
\textsc{intercept}, the constant; \textsc{slope}, the score; \textsc{link}, three members, $(\eta_0^2,\eta_0^3)$, the two
terms of \citet{stukel1988}, and a natural spline in $\eta_0$ with four degrees of freedom; \textsc{cov}, three members,
the squares and cubes of the continuous covariates, their pairwise products, and natural splines in each covariate.
Each basis is projected, in the weighted inner product, off the earlier parts: the link bases off $\{1,\eta_0\}$, the
covariate bases off a five-degree-of-freedom spline in $\eta_0$. A member is the usual score statistic
$Q=U^{\top}V^{-1}U$ with $U=\sum_i\tilde b(X_i)\{Y_i-p_0(X_i)\}$ and $V=\sum_iw_0\tilde b\tilde b^{\top}$, and its $\chi^2$
$p$-value $P$. Each member is the locally optimal test in its own directions (Proposition~\ref{prop:optimal}); the group spends some
of that optimality to cover several directions.

In-sample, the \textsc{intercept} and \textsc{slope} groups are dropped, the bases are evaluated at the fitted score and
also projected off the model matrix, and $Y_i-\hat p_i$ replaces $Y_i-p_0(X_i)$.

\subsection{Combination, closure and calibration}

For any set $S$ of groups, the statistic $T_S$ is the mean of $\tan\{(1/2-P)\pi\}$ over the members of the groups in
$S$, the Cauchy combination of \citet{liuxie2020}. A group is named when $T_S$ is significant for every set $S$ that
contains it; this is the closure principle of \citet{marcus1976closed}, and with four groups it needs fifteen tests.
Holm's or Bonferroni's correction applied to the groups' own tests is the closure with Bonferroni intersections; it
carries the same guarantee, is offered as an option, and is compared in Section~\ref{s:sim}.
Each $T_S$ is calibrated by its rank among reference draws of the whole procedure:
\begin{itemize}
\item[] external validation: $Y^{*}_i\sim\mathrm{Bernoulli}\{p_0(X_i)\}$, the Monte Carlo test of
  \citet{barnard1963discussion}, exact under no misfit at any sample size;
\item[] in-sample: $Y^{*}_i\sim\mathrm{Bernoulli}(\hat p_i)$ with the model refitted to each draw, the parametric bootstrap
  \citep{stute1998boot};
\item[] external validation, for theory only at present: a multiplier reference that flips the signs of the observed
  residuals, with each member studentized by its own residuals after refitting the directions it is projected off;
  it is valid for departures of any size asymptotically, but it was liberal in our samples (Section~\ref{s:sim}).
\end{itemize}
\emph{A robust option.} Squares, cubes and products of covariates give a few records with extreme covariate values
a large share of a statistic. With the robust option, every continuous covariate, every continuous model column and
the score are replaced by their normal scores, $\Phi^{-1}\{\mathrm{rank}/(n+1)\}$, before the bases are built, so each
record's contribution is bounded by a power of $\log n$. The pieces are then defined on the normal-score scale. In
external validation the Monte Carlo reference stays exact; in-sample the bootstrap still refits the model on the original
design, and the level was checked by simulation. One caution: the \textsc{slope} and \textsc{link} bases are then
separated against the normal score of $\eta$ rather than $\eta$ itself, so when the score is far from Gaussian a pure slope
error can partly appear as \textsc{link}, and the rung the robust option names can be one step too high.

Because the calibration is by simulation, the validity of the procedure does not rest on the Cauchy combination,
which is chosen for power. The procedure is implemented as \texttt{localize.external()} and \texttt{localize.gof()} in the
R package ebrahim.gof (development version, to be released on CRAN) and as a standalone script in the deposit; with
$M=999$ draws an external analysis of 1,000 patients takes under a second.

\subsection{Guarantees}

The results are asymptotic for a fixed number of covariates, under local departures of size $n^{-1/2}$, and are
stated for external validation; the proofs are in Appendix~A. In-sample, items 1 and 2 hold for \textsc{cov}, and for
\textsc{link} they hold for the de-aliased group or when the average of the covariates given the score is linear
(Section~2.2).
\begin{enumerate}
\item \emph{The parts are read separately} (Theorem~\ref{thm:law}). Each member's noncentrality depends only on its own part of
the misfit, and the groups are asymptotically independent.
\item \emph{Strong familywise error control} (Theorem~\ref{thm:fwer}). Whatever the other parts are, the probability of naming
any part whose misfit is zero is at most the level $\alpha$.
\item \emph{Consistency} (Theorem~\ref{thm:consist}). A part that is present along its group's bases is named with
probability tending to one.
\item \emph{The decision rule} (Corollary~\ref{cor:rung} and Proposition~\ref{prop:update}). The highest part named equals the highest part
present with probability at least $1-\alpha$ as the misfit grows, and updating at that rung removes the systematic misfit,
up to estimation noise.
\item \emph{Any prediction model} (Corollary~\ref{cor:anymodel}). In external validation the predictions may come from any model fitted
elsewhere, including machine-learning models, and the covariates may be any recorded in the validation sample.
\item \emph{Departures of any size} (Theorem~\ref{thm:multiplier}). With the multiplier reference, the error control holds for departures
of fixed size, stated on the probability scale; this is asymptotic, and at $n\le4{,}000$ the multiplier reference was
liberal (Section~\ref{s:sim}).
\end{enumerate}
The local qualification in 2--4 matters. A large departure that is pure on the logit scale is not pure on the
probability scale: a strong slope error also bends the probability curve, so naming \textsc{link} as well is then a
correct detection on the probability scale. The variance also changes away from the null, which inflated the error by
up to about 13 points in our fixed-size departures (Proposition~\ref{prop:beyond}). When predicted risks below 0.5 are
too high, as after the transport studied in Section~\ref{s:app}, the effect runs the other way: the model variance
overstates the true one, and the Monte Carlo tests of the other parts become conservative and lose power.

\subsection{Reading a verdict}

\begin{table}
\caption{Reading a verdict: act on the highest part named. In-sample only \textsc{link} and \textsc{cov} can be named (Section~\ref{s:parts}).}
\label{t:decision}
\begin{center}\footnotesize
\begin{tabular}{p{1.5cm}p{3.6cm}p{4.2cm}p{4.6cm}}
\Hline
Highest part named & What it means & Typical cause & Update \\ \hline
none & no misfit detected &  & keep the model; the verdict is limited by the sample size (Corollary~\sref{cor:power}) \\
\textsc{intercept} & overall risk too high or too low & different prevalence, case mix or period & update the intercept \\
\textsc{slope} & risks too extreme or too modest & overfitting at development; a different spread of risk & logistic recalibration $a+b\eta_0$; shrink future models \\
\textsc{link} & curved relation between score and risk & asymmetric tails; a transformation missing from the score & flexible recalibration $f(\eta_0)$ or another link \\
\textsc{cov} & misfit among patients who share a score & omitted interaction, curved covariate, an effect that differs in the new population & revise the model; recalibration cannot repair it \\
\hline
\end{tabular}
\end{center}
\end{table}

Table~\ref{t:decision} turns the verdict into an action. The rule is to act on the highest part named. When
\textsc{cov} is named, the model must be revised, even if other parts are named too, and the remaining parts are
assessed again after the revision. In-sample, a \textsc{link} verdict without \textsc{cov} says that adding a function of
the score to the model, or changing the link, is consistent with the data; it is not proof that no revision is
needed. When the miscalibration is gross, as it often is after transport to another population, the Monte Carlo verdict
may miss \textsc{cov}, and its named rung may be too high, because a large logit-scale error also bends the probability
curve. We then recommend a split-first update-and-retest: split the validation data at random into halves A and B,
reach the verdict on A, update the model on A at the rung it names (preferring the lighter update when
\textsc{intercept} and \textsc{slope} are also named), and test the updated model on B. On B the updated model is a frozen
predictor, so the Monte Carlo reference is exact under its null; the update's own estimation error lies in the
directions it fitted (Proposition~\ref{prop:update}), so it cannot create a \textsc{cov} verdict, and the parts above the
update rung are tested with approximately the local guarantees. Parts at or below the rung carry that estimation error,
a genuine misfit of the updated model, and may be named on B. Before a
validation study, Corollary~\ref{cor:power} gives the smallest miscalibration each group can detect, from the
covariates and the frozen predictions alone; for example, a mean miscalibration $\delta$ is detected with power $0.80$
at level $0.05$ when $n\ge7.85\,\E(w_0)/\delta^2$.

\section{Simulation Studies}
\label{s:sim}

\subsection{Design}

Five independent standard normal covariates and the score $\eta_0=-0.5+0.5x_1+0.8x_2+0.6x_3+0.4x_4+0.3x_5$; the true
model is $\logit\,\pi=\eta_0+d(x)$ with the departures $d$ of Figure~\ref{f:sim}: an intercept shift; a slope error;
a curved link, $\eta_0^2-1$; a U-shaped covariate effect, $x_1^2-1$; a threshold, $\max(x_2-1,0)$; and an omitted
interaction, $x_1x_3$. In external validation the frozen model is $\eta_0$; in-sample the analyst fits the linear model
in $x_1,\dots,x_5$. Each cell has 300 data sets (1,000 under the null) of size 500 or 1,000; $M=999$ Monte Carlo draws
or $B=199$ bootstrap draws; $\alpha=0.05$. We report how often each group is named and the \emph{right-action rate}, the
proportion of data sets whose highest named group is the rung the departure needs.

\subsection{Results}

\begin{figure}
\centerline{\includegraphics[width=\textwidth]{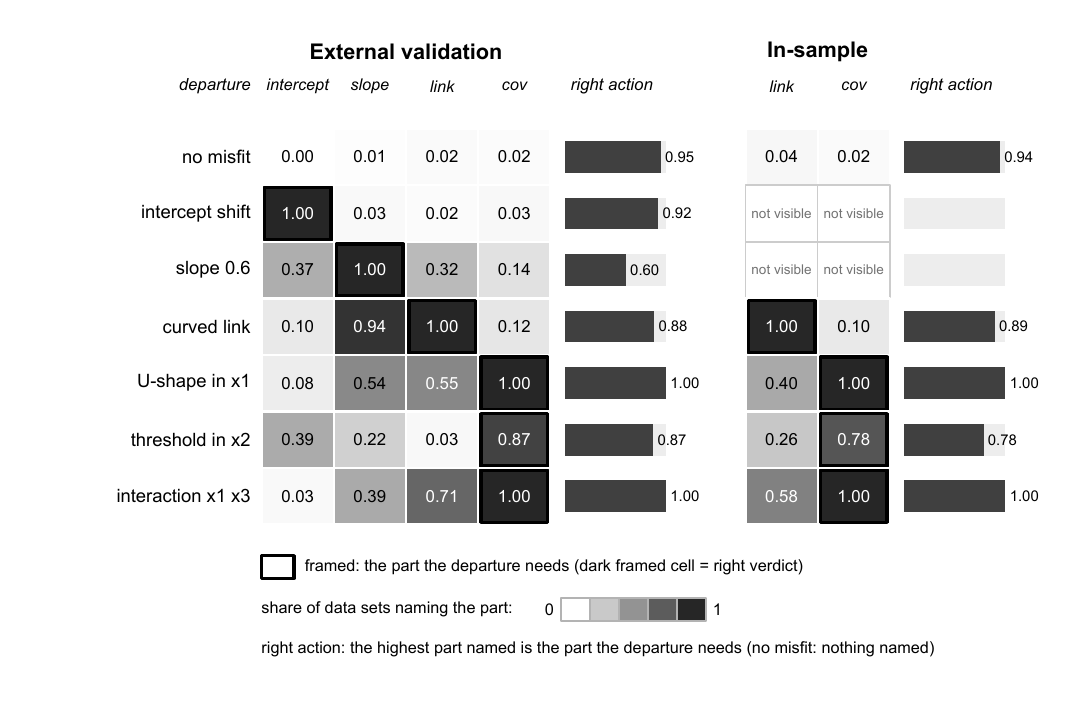}}
\caption{Simulation, $n=1{,}000$: share of data sets in which each part was named (shade and number), with the part each
departure needs framed, and the right-action rate (bars: the highest part named is the part the departure needs; under
no misfit, nothing named). External validation: Monte Carlo reference, $M=999$; in-sample: parametric bootstrap,
$B=199$. In-sample only \textsc{link} and \textsc{cov} can be seen, so the intercept and slope departures were run
externally only. 300 data sets per row, 1,000 under no misfit.}
\label{f:sim}
\end{figure}

Under no misfit the familywise error was $0.037$ and $0.051$ externally and $0.049$ and $0.055$ in-sample at $n=500$ and
$1{,}000$. Figure~\ref{f:sim} shows the results at $n=1{,}000$. Covariate-structure departures were located correctly:
\textsc{cov} was named in 87--100\% of data sets externally and 78--100\% in-sample, and the right action was chosen in all
of them that were detected. The
curved link was named in all data sets, and the right action was chosen in 88--89\%; the remainder also named
\textsc{cov}. The intercept shift was named in all data sets and the right action chosen in 92\%. The slope error is the
instructive case: \textsc{slope} was named in every data set, but \textsc{link} was named too in a third of them and
\textsc{cov} in 14\%. The first is largely a real effect: a slope error of this size on the logit scale bends the
probability curve, which gives the link part a non-zero probability-scale component. The second is not: \textsc{cov} has
no such component here, and its excess over the nominal level is the variance effect of
Proposition~\ref{prop:beyond}.

\emph{Comparisons.} On the same data sets, the procedure's global test (the intersection of all groups) detected every
departure within 0.01 of the best single test, and far more often than the partition tests and the calibration belt
when the misfit lay within patient groups (Table~\ref{tab:compare}; a part is named in marginally fewer data sets
than the global test rejects, for example 0.957 against 0.963 for the external threshold). Externally it detected the U-shape, threshold and interaction in 96--100\% of
data sets, against 46--58\% for the Hosmer--Lemeshow test and 53--77\% for the calibration belt; in-sample it detected
them in 83--100\%, against 16--28\% for the Hosmer--Lemeshow test. In-sample, the test of le Cessie and van Houwelingen matched it
against the U-shape and the interaction but not against the threshold (0.57 against 0.83) or the curved link (0.84
against 1.00), and it cannot say where the misfit lies. Reading the groups' own tests without error control named a
part wrongly in 16--18\% of data sets externally and 10--11\% in-sample when the model was right (Web
Table~\ref{tab:multiplicity}), although its right-action rate under misfit was no lower than the closure's. Holm's or Bonferroni's correction across the groups' own tests controlled the error as well
as the closed procedure and chose the right action slightly more often (0.85--0.86 against 0.83 externally at
$n=1{,}000$); in these designs the Cauchy closure showed no advantage over them, and it may gain only when a weak
part needs several tests combined. At $n=200$ and 300 the error
stayed at 0.038--0.053 and the power fell as expected (Table~\ref{tab:small}).

\begin{figure}
\centerline{\includegraphics[width=\textwidth]{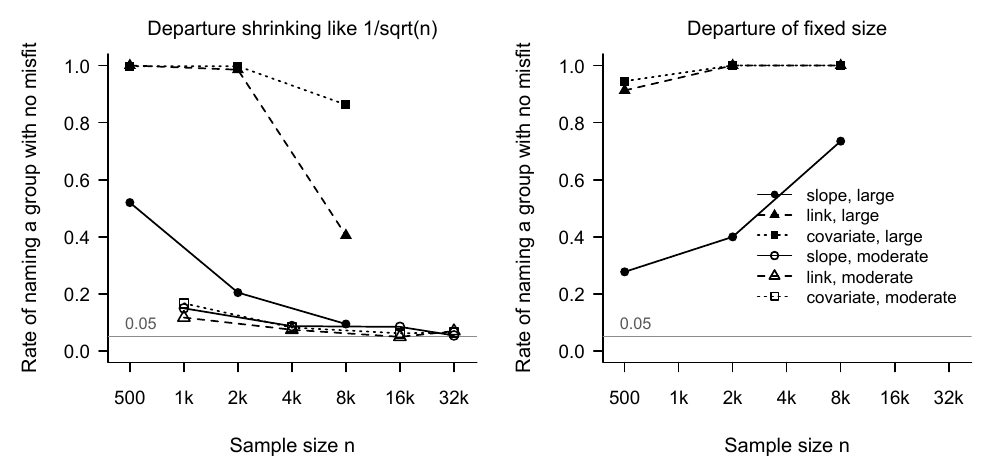}}
\caption{Rate of naming a group whose part of the misfit is zero, external validation, Monte Carlo reference. Left:
departures that shrink like $1/\sqrt n$, the setting of the error-control theorem; large departures (filled symbols) and
moderate ones followed to $n=32{,}000$ (open symbols). Right: the large departures held at a fixed size. The grey line is
the nominal level $0.05$. Each point is 400 data sets (200 at $n=32{,}000$).}
\label{f:local}
\end{figure}

Two further studies check the theory directly (Figure~\ref{f:local}; Appendix~A, Section~\ref{sec:numerics}). When
departures shrink like $n^{-1/2}$,
the rate of naming a group with no misfit fell toward the nominal level, from $0.12$--$0.17$ at $n=1{,}000$ to $0.065$
(standard error $0.006$) at $n=16{,}000$ and $32{,}000$, while the right group was still found in 70--100\% of data sets.
When the departures stayed at a fixed size, the same rate rose with $n$, as Proposition~\ref{prop:beyond} predicts.
Computing each group's true probability-scale part from the true risks showed that most of these extra namings were
real detections, and that the rest, up to about 13 points above the member's own power, was the variance effect.

The multiplier reference was checked on departures of fixed size that are pure on the probability scale, and on a
gross miscalibration of the kind met after transport, $\logit\,\pi=-2+0.8\eta_0$ (about 11\% observed against 40\%
predicted), which has no covariate-structure part (Table~\ref{tab:robust}). Under no misfit its error was $0.090$ at
$n=1{,}000$ and $0.070$ at $n=4{,}000$ (standard errors $0.014$ and $0.013$), against $0.045$ and $0.055$ for the Monte Carlo
reference, and under the gross miscalibration it named \textsc{cov} wrongly in $0.120$ and $0.085$ of data sets, against
$0.003$ and $0.000$ for the Monte Carlo reference. The multiplier reference is therefore liberal at these sample sizes,
exactly where it was meant to help, while the Monte Carlo reference is conservative there. We use the Monte Carlo
reference throughout, and the split-first update-and-retest of Section~\ref{s:method} when the miscalibration is gross.

\emph{Corrupted records.} When a covariate was multiplied by 4 or 8 in 3 of 500 or 10 of 1,000 records of a correctly
specified model, the default procedure named a part in 46--65\% of data sets, against 2--7\% on clean data; with the
robust option the rate fell to 22--42\% (Table~\ref{tab:robustopt}). On the main departures the robust option cost
almost nothing: its error was 0.046--0.047, its detection matched the default in every cell but the in-sample
threshold (0.79 against 0.83; right action 0.72 against 0.78), and its right-action rate was within 0.02 elsewhere.
What remains of the false verdicts likely comes from the fit: the corrupted records also pull the fitted model, which a
robust option for the tests cannot undo.

The de-aliased link group was studied in-sample with a skewed and a binary covariate, where the average covariate
given the score is curved (Appendix~B). Its error was at the nominal level, and it reduced the rate at which a
covariate-structure misfit was named \textsc{link} from $0.055$ to $0.030$ (standard errors $0.016$ and $0.012$, 200 data
sets). The price predicted by Proposition~\ref{prop:dealias} was large in this design: the power of \textsc{link} against a
curved link fell from $0.66$ to $0.010$, because the curvature in the score is almost entirely carried by the covariates, so that in-sample the
two explanations cannot be told apart. We therefore keep the default and use the de-aliased group as a sensitivity
check: a \textsc{link} verdict that survives it is a link verdict, and one that disappears may be covariate structure.

\section{Applications}
\label{s:app}

\subsection{Mortality in intensive care}
\label{s:support}

\begin{figure}
\centerline{\includegraphics[width=\textwidth]{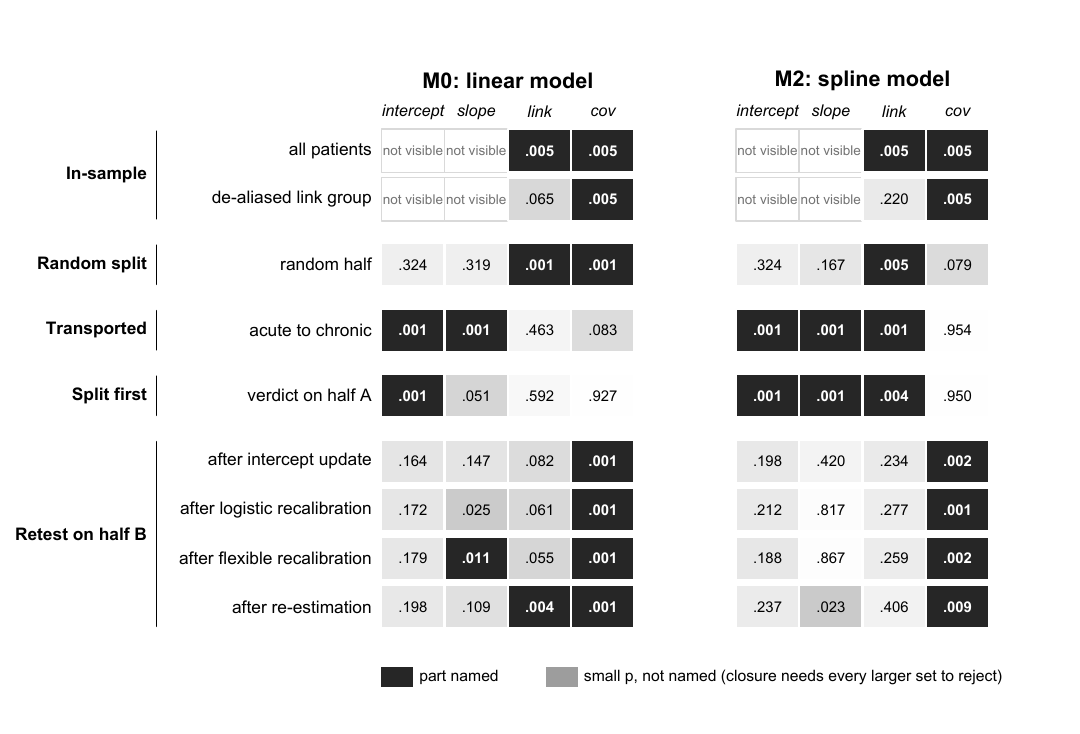}}
\caption{SUPPORT verdicts for the linear model M0 and the spline model M2: the $p$-value of each part's own test, with
the parts named ($\alpha=0.05$) in black. In-sample ($n=8{,}873$, bootstrap reference), on a random half after
development on the other half ($n=4{,}437$), transported from the acute to the chronic diagnoses ($n=4{,}130$), and the
split-first update-and-retest: verdict on a random half A of the chronic patients, updates fitted on A, retests on the
other half B ($n=2{,}065$). Monte Carlo reference for every external analysis. A part can have $p\le0.05$ and not be
named, because naming also needs every larger set of parts to reject.}
\label{f:pvalues}
\end{figure}

SUPPORT followed seriously ill hospitalized adults in five American medical centers \citep{knaus1995support}; we
use the 8,873 patients with complete data and in-hospital death (26\%) as outcome \citep{support2data}. Model M0 is
linear in age, mean arterial pressure, heart rate, respiratory rate, temperature, creatinine, sodium, white cell
count, number of comorbidities and sex; model M2 replaces the seven physiological variables by natural splines with
three degrees of freedom. We check each model in three settings: in-sample on all patients; externally on a random
half after development on the other half; and transported, developed on the acute diagnoses (acute respiratory
failure or multiple organ failure, coma) and validated on the chronic ones (heart failure, chronic lung disease,
cirrhosis, cancer). Figure~\ref{f:pvalues} shows the verdicts.

\emph{In-sample.} Both models fail, and \textsc{link} and \textsc{cov} are both named. SUPPORT has a binary covariate
(sex) and skewed laboratory values, so part of a covariate-structure misfit can appear as a curved link in-sample
(Section~\ref{s:parts}). With the de-aliased link group \textsc{link} is no longer named for either model ($p=0.065$ for
M0, $0.220$ for M2). Because de-aliasing can remove most of the power against a real link misfit (Section~\ref{s:sim}),
this is consistent with either explanation; external validation on similar patients, below, decides it.

\emph{Random split.} On similar patients the verdicts are the in-sample ones made precise. M0 needs revision
(\textsc{link} and \textsc{cov} named, $p\le0.001$). For M2 only \textsc{link} is named ($p=0.005$; \textsc{cov} $p=0.079$): the
splines removed most of the misfit within patient groups and left a curved relation between score and risk, which a
flexible recalibration repairs; this also points to a real link part in the in-sample verdict for M2. The overall risk is right (26.0\% observed, 26.6\% predicted) and neither \textsc{intercept}
nor \textsc{slope} is named.

\emph{Transported.} On the chronic diagnoses the models predict 40.0\% (M0) and 37.7\% (M2) mortality against 10.4\%
observed. The Monte Carlo reference names \textsc{intercept} and \textsc{slope} for M0, and \textsc{link} as well for M2.
This miscalibration is far outside the local regime, and the model variance, computed from risks near 0.40,
overstates the true variance of outcomes with risks near 0.10, so the covariate group is conservative and loses
power (Section~\ref{s:sim}). We therefore applied the split-first update-and-retest (Figure~\ref{f:pvalues}, lower rows).
On a random half A of the chronic patients the verdict was \textsc{intercept} for M0 (\textsc{slope} $p=0.051$) and
\textsc{intercept}, \textsc{slope} and \textsc{link} for M2. After the update each verdict prescribes, fitted on A, the
test on half B named \textsc{cov} for both models ($p=0.001$ and $0.002$), and \textsc{cov} was named on B after every
recalibration. A recalibration estimates only a function of the score, so its own estimation error lies in the
functions of the score, which \textsc{cov} cannot read: the \textsc{cov} verdicts on B are misfit of the transported
models, which the first check had missed. Re-estimating the same model on A also left \textsc{cov} named ($p=0.001$ for
M0, $0.009$ for M2), but there the verdict cannot be read as a defect of the model's form, because a model re-estimated
from about 225 deaths carries estimation error in the directions of its covariates, which \textsc{cov} can read. The
other parts named on B after an update (\textsc{slope} after the flexible recalibration of M0, \textsc{link} after its
re-estimation) are of the same kind: genuine misfit of the updated model that comes from its own estimation error. Figure~\ref{f:compass} shows the verdicts for M0 as the R function draws them: before the update only
\textsc{intercept} is named, and \textsc{slope} is just missed ($p=0.051$); after it, every set of parts that contains
\textsc{cov} rejects, so \textsc{cov} is named. A calibration plot cannot show this: after every update the calibration
curve on half B lies close to the diagonal (Figure~\ref{fig:calib}).

\providecommand{\compassfile}{fig_compass_plot}
\begin{figure}
\centerline{\includegraphics[width=\textwidth]{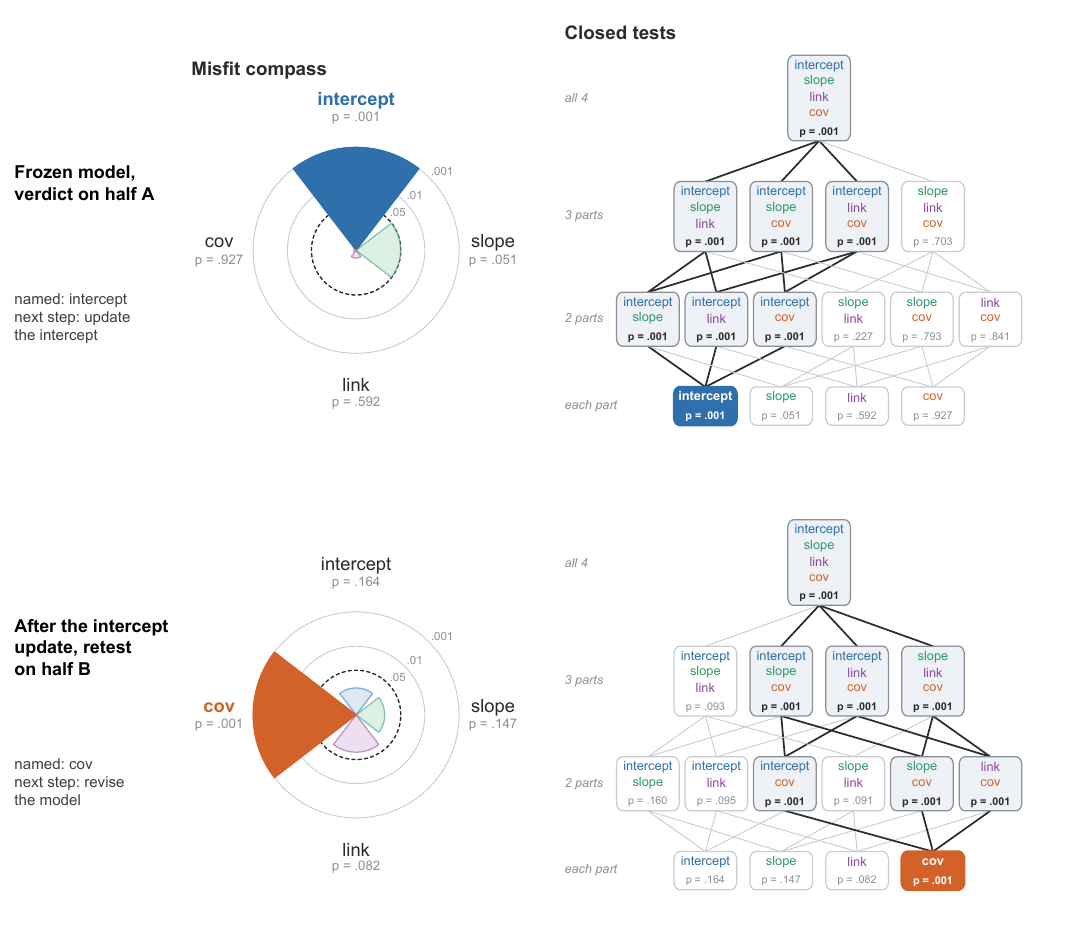}}
\caption{The verdict display of the R function, for the transported SUPPORT model M0: the frozen model on half A of
the chronic patients (top) and the model after the intercept update, retested on half B (bottom). Left, the misfit
compass: one wedge per part, whose length is $-\log_{10}p$ of the part's own test (capped at $p=0.001$); the dashed
ring is $p=0.05$, and the parts named are filled. Right, the closed tests: every set of parts with the $p$-value of its
test, shaded when it rejects; a part is named when every set that contains it rejects, and the links from a named part
up to those sets are drawn dark. Monte Carlo reference, $M=999$; $n=2{,}065$ in each half.}
\label{f:compass}
\end{figure}

\emph{Any prediction model.} A neural network with five hidden units, developed on the same data, was named
\textsc{slope} and \textsc{cov} on the random split (both $p\le0.001$), and \textsc{intercept} and \textsc{slope} when
transported (both $p\le0.001$; \textsc{link} $p=0.114$, \textsc{cov} $p=0.909$). Before any outcome is seen, Corollary~\ref{cor:power} gives the smallest miscalibration
each analysis can detect with power 0.80: a mean error of about two percentage points and a logit-scale slope error
of 0.14--0.31 (Table~\ref{tab:plan}); the transported miscalibration is about fourteen times the first.

\subsection{Mortality after burn injury}
\label{s:burn}

The second study is a genuine multi-centre validation at a clinical sample size: 1,000 patients treated for burns
in 40 facilities, 150 of whom died \citep[the burn data of][]{hosmer2013}. Model B0 has main effects for age, total
burn surface area, sex, race, inhalation injury and flame injury (Table~\ref{tab:burn}).

\emph{In-sample.} B0 is named \textsc{cov} ($p=0.005$); \textsc{link} is not ($p=0.075$), and stays unnamed with the
de-aliased group, so the verdict is a misfit within patient groups. Following it, a spline in age and the
interaction of age with inhalation injury improved the fit but left \textsc{cov} named; adding a spline in burn
surface area removed it (nothing named, $p=0.51$ and $0.34$; AIC from 350.5 to 333.2). Because this revision was
chosen by looking at the same data, its clean verdict is optimistic.

\emph{Across centers.} We developed each model on a random half of the facilities and validated it, frozen, on the
patients of the other half, about 500 patients and 74 deaths per validation set on average, over 20 random facility splits.
For B0, \textsc{cov} was named in 70\% of splits and nothing in 30\%: at this size one split can miss a real misfit,
as Corollary~\ref{cor:power} warns, and a verdict of ``nothing named'' is not proof of fit. For the revised model B3,
\textsc{cov} was still named in 75\% of splits. Part of this is the development half's own estimation error, which
\textsc{cov} can read. To measure it we simulated outcomes from each model fitted to all 1,000 patients, so that its
form is right by construction, and put 25 simulated data sets through the same 20 splits (Web
Table~\ref{tab:burn}). \textsc{cov} was then named in 17\% of splits for B0 (standard deviation 12\% across data
sets, largest 45\%) and 45\% for B3 (13\%, largest 70\%); none of the 25 data sets reached the observed rate for
either model. For B0 the observed 70\% is far above what estimation error gives, so the cross-centre misfit is real:
the covariate effects differ between centers or the model's form is wrong where the other half of the facilities
puts its patients. For B3 much of the rate is estimation error, because splines and an interaction estimated from
about 75 deaths per development half carry much of it; the observed rate still exceeds every simulated one, so some
cross-centre misfit remains, but the evidence for it is modest. In-sample fit
and transportability are different questions, and the procedure keeps them apart.

\section{Discussion}
\label{s:discuss}

A goodness-of-fit test tells the analyst that a risk model fails; the procedure tells where, with a guarantee on the
error, and the place determines the update. Three findings are worth carrying beyond this paper. The development
data cannot reveal an error in the overall risk or in the calibration slope, however large; an analyst who wants to
know whether a model's risks are too extreme must hold out data. Partition tests that group by fitted risk cannot see
misfit that is orthogonal to every function of the risk; a model that passes the Hosmer--Lemeshow test may still miss
an interaction.
And the same model can need different updates in different populations, as the transported SUPPORT model showed.

The battery is not the only possible one. Any test whose statistic depends on the data through the residuals can be
a member once its basis is projected into the part it should read; the partition tests read functions of the score and belong to \textsc{link} externally, and covariate-grouped tests such as those of \citet{tsiatis1980} and
\citet{xie2008} belong to \textsc{cov} once projected. Tests that read all directions at once, such as the projection
test of \citet{escanciano2006} and \citet{liu2024} or the learned test of \citet{zhang2023}, belong to no single part;
they can serve as a global detection step, but their cost is high: the projection test took nearly ten hours on the
SUPPORT data (Appendix~B), against about a minute for the whole in-sample procedure and a second for an
external validation here. A neural test that maps a model's failure over pairs of covariates \citep{deepgof2026} can
follow a \textsc{cov} verdict and show which covariates carry it. Like every score test, the default procedure reads a few records with extreme covariate
values as structure. The robust option, with bases on normal scores, cuts the resulting false verdicts by a third
to a half at
almost no cost in power (Section~\ref{s:sim}), and we recommend it whenever the covariates have not been checked for
gross errors; removing the rest needs a robust fit of the model as well, a topic for further work.

The theory is for a fixed number of covariates and the maximum-likelihood fit. Lemma~\ref{lem:erase} and the Monte Carlo
reference hold in any dimension, but when the number of covariates is not small relative to $n$ the maximum-likelihood
estimate is biased and its variance underestimated \citep{sur2019}, so the in-sample bootstrap is no longer justified,
and penalized fits do not satisfy the identities that remove the intercept and slope parts; for ridge fits even the
Hosmer--Lemeshow test loses its level \citep{ebrahim2026shrinkage}. Directed tests for
high-dimensional models \citep{jankova2020goodness} suggest one route.

\section*{Acknowledgements}

The authors used a large language model to improve the language of the manuscript and to assist with the
implementation of the simulation code and with the derivations; the authors checked every result, reviewed and edited
all content, and take full responsibility for it. The authors received no specific funding for this work and declare
no conflicts of interest.

\section*{Data Availability Statement}

The SUPPORT data are public \citep{support2data}, and so are the burn data \citep{hosmer2013}, distributed with the R
package aplore3. The R function implementing the procedure, the simulation code, the
per-data-set results behind every table, and the scripts that regenerate them are archived at Zenodo
(\url{https://doi.org/10.5281/zenodo.23192535}) and supplied with this paper.

\bibliographystyle{plainnat}
\bibliography{references}

\section*{Supplementary Material}

Appendix~A (theory and proofs, with numerical checks of the theory) and Appendix~B (simulation and application details)
follow the references.

\clearpage
\appendix
\part*{Appendix A: Theory of the localization procedure}
\setcounter{section}{0}\setcounter{equation}{0}\setcounter{table}{0}
\renewcommand{\thesection}{A.\arabic{section}}\renewcommand{\theequation}{A.\arabic{equation}}\renewcommand{\thetable}{A.\arabic{table}}

\noindent This appendix gives the theory behind Section~2 of the paper. Several ingredients are standard and are
stated so that the argument is self-contained: the score identity of the logistic fit (Lemma~\ref{lem:erase}); the
effect of estimated parameters on a score statistic (Pierce, 1982), which is Theorem~\ref{thm:insample}(a); the
pseudo-true first-order condition (White, 1982), which is Theorem~\ref{thm:insample}(c); the closure principle (Marcus
et al., 1976); and Monte Carlo and parametric-bootstrap calibration (Dufour, 2006; Stute et al., 1998). The
observation behind Corollary~\ref{cor:hl}, that statistics indexed by the fitted score cannot see departures orthogonal
to it, also motivates the projection tests of Stute and Zhu (2002) and Escanciano (2006). Further standard
consequences are stated for use: external score tests apply to predictions from any model (Corollary~\ref{cor:anymodel}),
the usual power calculation for observed-versus-expected tests (Corollary~\ref{cor:power}), the local optimality of
score tests (Wald, 1943; Proposition~\ref{prop:optimal}), and multiplier (wild-bootstrap) calibration (Wu, 1986;
Mammen, 1993), used in Theorem~\ref{thm:multiplier}. What is new is the combination: the four-piece decomposition
of a misfit and its link to the calibration hierarchy (Definition~\ref{def:pieces}); the exact description of what
in-sample checking can and cannot separate, including when the link group is aliased with covariate structure
(Theorem~\ref{thm:insample}, Proposition~\ref{prop:alias}) and how to remove the aliasing (Proposition~\ref{prop:dealias});
the law of the group statistics (Theorem~\ref{thm:law}); the guarantees of the closed localization procedure
(Theorems~\ref{thm:fwer} and~\ref{thm:consist}, Corollary~\ref{cor:rung}) and their extension to departures of any size
(Theorem~\ref{thm:multiplier}); the link between a verdict and the model update that repairs it
(Proposition~\ref{prop:update}); and a covariate group that becomes exact in large samples (Proposition~\ref{prop:sieve}).
Proposition~\ref{prop:beyond} describes what happens outside the local regime. All results are for a fixed number of
covariates and the unpenalized logistic maximum-likelihood fit; Remark~\ref{rem:scope} gives the scope.

\section{Setting}

\paragraph{Data and model.} The data are independent copies $(X_i,Y_i)$, $i=1,\dots,n$, of $(X,Y)$ with
$X\in\mathbb{R}^d$ and $Y\in\{0,1\}$; write $\tx=(1,x^{\top})^{\top}$ and $\pi(x)=\Pr(Y=1\mid X=x)$. The working model is
logistic, $\mathrm{logit}\,p_\beta(x)=\tx^{\top}\beta$. Two settings are studied.
\begin{itemize}[leftmargin=*,itemsep=1pt]
\item \emph{External validation.} A model fitted elsewhere gives a fixed linear predictor $\eta_0(x)=\tx^{\top}\beta_0$
  and risks $p_0=\expit(\eta_0)$; nothing is re-estimated on the validation data.
\item \emph{In-sample checking.} $\hat\beta$ is the maximum-likelihood estimate from the same data, $\hat p_i =
  \expit(\tx_i^{\top}\hat\beta)$ and $\hat\eta_i=\tx_i^{\top}\hat\beta$; $\beta_0$ denotes its limit.
\end{itemize}

\paragraph{Inner product.} With $w_0=p_0(1-p_0)$, let $L^2(w)$ be the functions $f$ of $x$ with
$\nrm{f}^2=\E\{w_0(X)f(X)^2\}<\infty$ and $\ip{f}{g}=\E\{w_0(X)f(X)g(X)\}$. For a closed subspace $\cH$, $\Pi_{\cH}$ is
the orthogonal projection onto $\cH$ in $L^2(w)$ and $\cH^{\perp}$ its orthogonal complement.

\paragraph{The probability residual.} Misfit is measured on the probability scale. For a basis function $b$ define
\begin{equation}\label{eq:g}
\ip{b}{g}:=\E\{b(X)(\pi(X)-p_0(X))\}=\E\{b(X)(Y-p_0(X))\},\qquad g=\frac{\pi-p_0}{w_0},
\end{equation}
the second form being meaningful whenever $\E\{(\pi-p_0)^2/w_0\}<\infty$. A score statistic built on $b$ reads the
misfit only through $\ip{b}{g}$. Under the \emph{local alternatives}
\begin{equation}\label{eq:local}
\mathrm{logit}\,\pi_n(x)=\eta_0(x)+n^{-1/2}h(x),
\end{equation}
Taylor's theorem with $|\expit''|\le 2c_2$, $c_2=1/(12\sqrt3)$, and with $|\expit''(\eta+t)|\le w(\eta+t)\le e^{|t|}w(\eta)$ gives
\begin{equation}\label{eq:rem}
|\pi_n-p_0-n^{-1/2}w_0h|\le n^{-1}h^2\min\{c_2,\ \tfrac12e^{|h|/\sqrt n}w_0\} .
\end{equation}
The first bound is used in all sums below (it needs only polynomial moments); the second shows pointwise that
$g_n=n^{-1/2}h+O(n^{-1}h^2e^{|h|/\sqrt n})$: $h$ is the probability residual to first order, and the two scales
agree locally. They differ for large departures (Proposition~\ref{prop:beyond}).

\section{The four parts of a misfit}

\begin{definition}[Pieces]\label{def:pieces}
Let $\cG\subset L^2(w)$ be the closed subspace of functions of $\eta_0(x)$ alone, $\cX=\mathrm{span}\{1,x_1,\dots,x_d\}$,
and $\bar\eta=\E(w_0\eta_0)/\E(w_0)$.
\begin{center}
\begin{tabular}{llll}
\toprule
piece & space & meaning & level in Van Calster et al.\ (2016)\\
\midrule
$\cH_{\rm I}$ & $\mathrm{span}\{1\}$ & overall risk too high or low & calibration in the large\\
$\cH_{\rm S}$ & $\mathrm{span}\{\eta_0-\bar\eta\}$ & risks too extreme or too modest & weak (slope)\\
$\cH_{\rm L}$ & $\cG\ominus\mathrm{span}\{1,\eta_0\}$ & curved score-to-risk map (link shape) & moderate\\
$\cH_{\rm C}$ & $\cG^{\perp}$ & misfit within covariate patterns & strong\\
\bottomrule
\end{tabular}
\end{center}
The four pieces are mutually orthogonal and $L^2(w)=\cH_{\rm I}\oplus\cH_{\rm S}\oplus\cH_{\rm L}\oplus\cH_{\rm C}$, the
orthogonal sum along the nested spaces $\mathrm{span}\{1\}\subset\mathrm{span}\{1,\eta_0\}\subset\cG\subset L^2(w)$.
Write $h=h_{\rm I}+h_{\rm S}+h_{\rm L}+h_{\rm C}$ for the components of $h$, and similarly for $g$.
\end{definition}

\noindent In words, $h_{\rm I}+h_{\rm S}$ is the part a logistic recalibration $\mathrm{logit}\,\pi=a+b\eta_0$ repairs;
$h_{\rm L}$ is the part a flexible recalibration (any function of the score) repairs but a linear one does not; and
$h_{\rm C}$ is the part no function of the score repairs, so the model itself must change. These are the steps of the
model-updating strategies of Steyerberg et al.\ (2004) and Janssen et al.\ (2008), with flexible recalibration between
logistic recalibration and model revision.

\paragraph{Groups and members.} Group $k\in\{{\rm I},{\rm S},{\rm L},{\rm C}\}$ has $J_k$ members. Member $(k,j)$ is a
vector of $r_{kj}$ basis functions $b_{kj}$ projected onto piece $k$, $\tilde b_{kj}=\Pi_{\cH_k}b_{kj}$ componentwise.
The bases used in the paper are $1$ (I); $\eta_0$ (S); $(\eta_0^2,\eta_0^3)$, the two terms of Stukel (1988) and a natural
spline in $\eta_0$ (L); squares and cubes, pairwise products and natural splines of the covariates (C).

\begin{assumption}\label{ass:basis}
$h\in L^2(w)$; each $V_{kj}=\E\{w_0\tilde b_{kj}\tilde b_{kj}^{\top}\}$ is positive definite (redundant columns
removed; in-sample the same for the projected bases $\Pi_{\cK_k}b_{kj}$); and for some $\delta>0$,
$\E\{w_0\lVert\tilde b_{kj}\rVert^{2+\delta}\}<\infty$, $\E\{\lVert\tilde b_{kj}\rVert\,h^2\}<\infty$ and
$\E\{(\lVert\tilde b_{kj}\rVert^{2}+\lVert\tilde b_{kj}\rVert^{2+\delta})|h|\}<\infty$. These hold for the polynomial, product and spline bases and the departures
($\eta_0^2$, $x_1x_3$, \dots) used in the paper when the covariates have enough moments.
\end{assumption}

\paragraph{Statistics.} For external validation, $U_{kj}=n^{-1/2}\sum_i\tilde b_{kj}(X_i)\{Y_i-p_0(X_i)\}$,
$\hat V_{kj}=n^{-1}\sum_i w_0(X_i)\tilde b_{kj}(X_i)\tilde b_{kj}(X_i)^{\top}$, $Q_{kj}=U_{kj}^{\top}\hat V_{kj}^{-1}U_{kj}$
and $P_{kj}=1-F_{r_{kj}}(Q_{kj})$, with $F_r$ the $\chi^2_r$ distribution function. For a nonempty set $S$ of groups,
\begin{equation}\label{eq:TS}
T_S=\frac{1}{\sum_{k\in S}J_k}\sum_{k\in S}\sum_{j=1}^{J_k}\tan\{(\tfrac12-P_{kj})\pi\},
\end{equation}
the Cauchy combination of Liu and Xie (2020). Its reference law is obtained by Monte Carlo: $Y^{(m)}_i\sim
\mathrm{Bernoulli}\{p_0(X_i)\}$ independently, $m=1,\dots,M$, giving $T_S^{(m)}$, and
$p_S=\{1+\#(m:T_S^{(m)}\ge T_S)\}/(M+1)$. For in-sample checking, groups I and S are dropped, the bases are evaluated at
$\hat\eta$ and also projected off $\cX$, $Y_i-p_0(X_i)$ is replaced by $Y_i-\hat p_i$, and the reference draws are
$Y^{(m)}_i\sim\mathrm{Bernoulli}(\hat p_i)$ with the model refitted to each draw (the parametric bootstrap, $B$ draws).

\paragraph{Hypotheses and closed testing.} In external validation, $H_k$ states that group $k$ has nothing to see:
$\ip{\tilde b_{kj}}{h}=0$ for all $j$ (local version) or $\ip{\tilde b_{kj}}{g}=0$ for all $j$ (fixed version). Because
$\tilde b_{kj}\in\cH_k$, $\ip{\tilde b_{kj}}{h}=\ip{\tilde b_{kj}}{h_k}$: \emph{$H_k$ depends on the misfit only through its
own piece}, and $h_k=0$ implies $H_k$. The in-sample hypotheses are defined in Section~\ref{sec:insample}.
\emph{Group $k$ is named} when $p_S\le\alpha$ for every $S\ni k$ (Marcus et al., 1976).

\section{What in-sample checking can see}\label{sec:insample}

\begin{lemma}[Known]\label{lem:erase}
If the logistic model with an intercept is fitted by maximum likelihood and $\hat\beta$ exists, then
$\sum_i(Y_i-\hat p_i)\tx_i=0$. In particular $\sum_i(Y_i-\hat p_i)=0$ and $\sum_i(Y_i-\hat p_i)\hat\eta_i=0$, so, if
$\hat\eta$ is not constant, the logistic recalibration of $\hat p$ on the same data has intercept $0$ and slope $1$ exactly,
whatever $\pi$ is.
\end{lemma}
\begin{proof}
The score of the logistic likelihood, $\sum_i\tx_i(Y_i-p_\beta(X_i))$, vanishes at $\hat\beta$; take the first
coordinate and the combination $\hat\beta^{\top}\sum_i\tx_i(Y_i-\hat p_i)$. The recalibration model
$\mathrm{logit}\,\pi=a+b\hat\eta$ has score $\sum_i(1,\hat\eta_i)^{\top}\{Y_i-\expit(a+b\hat\eta_i)\}$, which is zero at
$(a,b)=(0,1)$; its log likelihood is strictly concave when $\hat\eta$ is not constant, so $(0,1)$ is the maximizer.
\end{proof}

\begin{assumption}\label{ass:fit}
$I_0=\E\{w_0\tx\tx^{\top}\}$ is positive definite, $\E\lVert X\rVert^{4}<\infty$ and $\E\{w_0\lVert\tx\rVert\,|h|\}<\infty$. The in-sample bases
$b(x;\theta)$ depend on $\theta=(\beta,\text{projection coefficients},\text{knots})$ continuously differentiably near
$\theta_0$, with $\E\sup_{\lVert\theta-\theta_0\rVert<\epsilon}\lVert\partial_\theta b(X;\theta)\rVert^2(1+\lVert X\rVert^2)<\infty$,
and $\hat\theta-\theta_0=O_p(n^{-1/2})$. (Polynomial bases need the corresponding moments of $X$; natural splines with
interior knots at estimated quantiles satisfy the condition; boundary knots are treated in Remark~\ref{rem:impl}.)
\end{assumption}

\begin{theorem}[In-sample checking sees the model's orthocomplement]\label{thm:insample}
Let $b$ be any basis function and $\tilde b=b-\Pi_{\cX}b$ its projection off $\cX$.
\begin{enumerate}[label=(\alph*),leftmargin=*,itemsep=2pt]
\item Under the local alternatives \eqref{eq:local} and Assumptions~\ref{ass:basis}--\ref{ass:fit},
\begin{equation}\label{eq:expansion}
n^{-1/2}\sum_i\hat b(X_i)(Y_i-\hat p_i)=n^{-1/2}\sum_i\tilde b(X_i)\{Y_i-\pi_n(X_i)\}+\ip{\tilde b}{h-\Pi_{\cX}h}+o_p(1),
\end{equation}
where $\hat b$ is $b$ or $\tilde b$ evaluated at the estimates. The statistic converges to
$N\big(\ip{\tilde b}{h-\Pi_{\cX}h},\,\E\{w_0\tilde b\tilde b^{\top}\}\big)$. Refitting removes $\Pi_{\cX}h$ from every such
statistic; its variance is that of the projected basis, so a quadratic form must use $\E\{w_0\tilde b\tilde b^{\top}\}$, not
$\E\{w_0bb^{\top}\}$ (the software projects first, so the two coincide).
\item Since $\mathrm{span}\{1,\eta_0\}\subset\cX$, the risk-scale part $h_{\rm I}+h_{\rm S}$ is invisible in-sample: every
test built from statistics of this form has local power equal to its size against any $h\in\cX$, in particular against
$h\in\mathrm{span}\{1,\eta_0\}$. What
remains, $h-\Pi_{\cX}h$, lies in $\cX^{\perp}=\cK_{\rm L}\oplus\cK_{\rm C}$ with
\[
\cK_{\rm L}=(\cG+\cX)\ominus\cX,\qquad \cK_{\rm C}=(\cG+\cX)^{\perp}=\cH_{\rm C}\cap\cX^{\perp}\subseteq\cH_{\rm C}
\]
($\cG+\cX$ is closed, being a closed subspace plus a finite-dimensional one).
\item The erasure is not only local. If $\pi$ is fixed and $\beta^*$ maximizes $\E\{Y\tx^{\top}\beta-\log(1+e^{\tx^{\top}\beta})\}$,
then $p^*=\expit(\tx^{\top}\beta^*)$ satisfies $\E\{\tx(\pi-p^*)\}=0$: the probability residual of the pseudo-true fit has
no component in $\cX$ (White, 1982).
\item More generally, refitting any smooth parametric model erases its tangent space: the in-sample statistics read
$h-\Pi_{\mathcal T}h$, with $\mathcal T$ the span of the derivatives of $\mathrm{logit}\,p_\beta$ in $\beta$. For the logistic model
$\mathcal T=\cX$. For a probit fit $\mathcal T=\mathrm{span}\{\tx\,\phi(\eta)/w\}$, which contains the probit-scale intercept and
slope directions but not the logit-scale ones, so the two identities of Lemma~\ref{lem:erase} fail (Table~\ref{tab:t1}),
yet what the refit erases is again the model's own risk-scale directions.
\end{enumerate}
\end{theorem}

\begin{proof}
(a) \emph{Step 1: the estimator.} Under \eqref{eq:local} the log likelihood ratio of $\pi_n$ against $p_0$ is
asymptotically normal with mean $-\frac12\nrm{h}^2$ and variance $\nrm{h}^2$ (local asymptotic normality of the
Bernoulli model, using $h\in L^2(w)$), so the alternatives are contiguous to $p_0$. Under $p_0$ the maximum-likelihood
estimator is asymptotically linear, $\sqrt n(\hat\beta-\beta_0)=I_0^{-1}n^{-1/2}\sum_i\tx_i(Y_i-p_0(X_i))+o_p(1)$; by
contiguity (Le Cam's first lemma) the remainder stays $o_p(1)$ under \eqref{eq:local}, and by the third lemma
$\sqrt n(\hat\beta-\beta_0)\dto N(I_0^{-1}\E\{w_0\tx h\},I_0^{-1})$, in particular $O_p(1)$.

\emph{Step 2: the projected basis at the true parameter.} Write
$Y_i-\hat p_i=\{Y_i-\pi_n(X_i)\}+\{\pi_n(X_i)-p_0(X_i)\}-\{\hat p_i-p_0(X_i)\}$. By Taylor's theorem
$\hat p_i-p_0(X_i)=w_0(X_i)\tx_i^{\top}(\hat\beta-\beta_0)+R_i$ with $|R_i|\le c_2\lVert\tx_i\rVert^2\lVert\hat\beta-\beta_0\rVert^2$, so,
using $\E\{\lVert\tilde b\rVert\lVert\tx\rVert^2\}<\infty$ (Cauchy--Schwarz with the moments assumed),
\[
n^{-1/2}\sum_i\tilde b(X_i)\{\hat p_i-p_0(X_i)\}=\Big\{n^{-1}\sum_iw_0\tilde b\tx^{\top}\Big\}\sqrt n(\hat\beta-\beta_0)+o_p(1).
\]
The bracket converges almost surely to $\E\{w_0\tilde b\tx^{\top}\}=0$, because $\tilde b\perp\cX$; with Step~1 the term is
$o_p(1)$. By \eqref{eq:rem} and Assumption~\ref{ass:basis},
$n^{-1/2}\sum_i\tilde b(X_i)\{\pi_n(X_i)-p_0(X_i)\}\to\ip{\tilde b}{h}$ almost surely, and $\ip{\tilde b}{h}=
\ip{\tilde b}{h-\Pi_{\cX}h}$ because $\tilde b\perp\cX$.

\emph{Step 3: an unprojected basis.} For a raw $b$ the same expansion gives
$n^{-1/2}\sum_ib(Y_i-\hat p_i)=n^{-1/2}\sum_ib(Y_i-p_0)-\{n^{-1}\sum_iw_0b\tx^{\top}\}\sqrt n(\hat\beta-\beta_0)+o_p(1)$;
inserting the linear representation of Step~1 and using $\E\{w_0b\tx^{\top}\}I_0^{-1}\tx=\Pi_{\cX}b$ gives
$n^{-1/2}\sum_i\tilde b(Y_i-p_0)+o_p(1)$, the statistic of $\tilde b$.

\emph{Step 4: estimated bases.} By Assumption~\ref{ass:fit},
$\lVert\hat b(x)-b(x;\theta_0)\rVert\le\sup_{\lVert\theta-\theta_0\rVert<\epsilon}\lVert\partial_\theta b(x;\theta)\rVert\,
\lVert\hat\theta-\theta_0\rVert$ on an event of probability tending to one. The contribution
$n^{-1/2}\sum_i\{\hat b(X_i)-b(X_i;\theta_0)\}\{Y_i-\pi_n(X_i)\}$ is the product of $\hat\theta-\theta_0=O_p(n^{-1/2})$ with an
$O_p(1)$ mean-zero sum, plus a remainder controlled by the same envelope (a standard stochastic-equicontinuity argument
for a finite-dimensional, differentiable class); and
$n^{-1/2}\sum_i\lVert\hat b-b\rVert\,|\pi_n-\hat p_i|=n^{1/2}O_p(n^{-1/2})O_p(n^{-1/2})=o_p(1)$. Natural splines place their
boundary knots at the sample range of $\hat\eta$, which has no fixed limit; centring instead at $\theta_{0,n}$, the same
basis with boundary knots at the sample range of $\eta_0$, the knot error is $O_p(n^{-1/2}\max_i\lVert\tx_i\rVert)$, which is
$O_p(\{\log n/n\}^{1/2})$ for Gaussian covariates, and the same two bounds apply.

\emph{Step 5: the limit.} Conditionally on $X_1,\dots,X_n$ the summands $\tilde b(X_i)\{Y_i-\pi_n(X_i)\}$ are independent with
mean zero and covariance $\tilde b\tilde b^{\top}\pi_n(1-\pi_n)$. Since $|\pi_n(1-\pi_n)-w_0|\le 2c_2\,n^{-1/2}|h|$ ($|w'|=|\expit''|$),
the average covariance converges almost surely to $\E\{w_0\tilde b\tilde b^{\top}\}$ by $\E\{\lVert\tilde b\rVert^2|h|\}<\infty$; and
since $\E\{|Y-\pi_n|^{2+\delta}\mid X\}\le\pi_n(1-\pi_n)\le w_0+2c_2n^{-1/2}|h|$, the Lyapunov condition follows from
$\E\{w_0\lVert\tilde b\rVert^{2+\delta}\}<\infty$ and $\E\{\lVert\tilde b\rVert^{2+\delta}|h|\}<\infty$. The
Lindeberg--Feller theorem gives the normal limit, which therefore also holds unconditionally.

(b) follows from (a): for $h\in\cX$ the mean is zero and the limit is the null limit. The
identity $(\cG+\cX)^{\perp}=\cG^{\perp}\cap\cX^{\perp}$ gives the last display.

(c) is the first-order condition of the expected log likelihood.

(d) The same argument applies with $\cX$ replaced by the tangent space $\mathcal T$, because Step~1 holds for any regular
parametric model and Step~2 uses only $\tilde b\perp\mathcal T$. For the probit model,
$\partial_\beta\,\mathrm{logit}\,\Phi(\tx^{\top}\beta)=\tx\,\phi(\eta)/\{\Phi(\eta)(1-\Phi(\eta))\}$.
\end{proof}

\paragraph{In-sample hypotheses.} In-sample the LINK members use $\Pi_{\cK_{\rm L}}b_{kj}$ and the COV members
$\Pi_{\cK_{\rm C}}b_{kj}$ (in the software: bases in $\hat\eta$ projected off $\{1,\hat\eta\}$ and the model matrix; covariate
bases also off a spline in $\hat\eta$). The in-sample hypothesis of group $k$ is
$H^{\rm in}_k:\ \ip{\Pi_{\cK_k}b_{kj}}{h}=0$ for all $j$, a statement about $\Pi_{\cK_k}h$ only.

\begin{proposition}[What the in-sample groups mean]\label{prop:alias}
\begin{enumerate}[label=(\roman*),leftmargin=*,itemsep=2pt]
\item In-sample COV reads only covariate structure: $\cK_{\rm C}\subseteq\cH_{\rm C}$, so $\Pi_{\cK_{\rm C}}h=\Pi_{\cK_{\rm C}}h_{\rm C}$ and
  a pure link departure ($h=h_{\rm L}$) never makes $H^{\rm in}_{\rm C}$ false.
\item In-sample LINK tests exactly whether adding a function of the score to the model improves it. The model
  $\mathrm{logit}\,\pi=\tx^{\top}\beta+\gamma^{\top}f(\tx^{\top}\beta)$ has tangent space $\cX+\mathrm{span}\{f(\eta_0)\}$ at $\gamma=0$, so its
  score test for $\gamma$ is the in-sample statistic of the basis $f$: the goodness-of-link test of Pregibon (1980), of
  which Stukel's (1988) test is the case $f=$ his two terms.
  Locally, $h$ is repaired to first order by such an extension with a rich enough $f$ if and only if $\Pi_{\cK_{\rm C}}h=0$.
\item Assume that no non-zero element of $\cX'=\cX\ominus\mathrm{span}\{1,\eta_0\}$ is a function of $\eta_0$ alone
  ($\cX'\cap\cG=\{0\}$). In-sample LINK coincides with the external link piece, that is $\cK_{\rm L}=\cH_{\rm L}$ and
  $\Pi_{\cK_{\rm L}}h_{\rm C}=0$ for every $h_{\rm C}\in\cH_{\rm C}$, if and only if $\E(\tx\mid\eta_0)$ is an affine function of
  $\eta_0$ for every column $\tx$ of the model matrix. This holds when the model-matrix columns are jointly elliptical,
  for example Gaussian main effects only. With binary covariates, interactions or nonlinear terms it generally fails,
  and in-sample LINK then also reads the part of a covariate-structure misfit that the refit turns into a function of
  the score; its verdict should be read through (ii).
\end{enumerate}
\end{proposition}

\begin{proof}
(i) is Theorem~\ref{thm:insample}(b). (ii) The tangent space of the extension is the derivative of the log odds in
$(\beta,\gamma)$ at $\gamma=0$, namely $\cX+\mathrm{span}\{f(\eta_0)\}$; the efficient score for $\gamma$ is the score projected
off the nuisance tangent space $\cX$, which is the in-sample statistic of $f$. A first-order repair exists if and only if
$h-\Pi_{\cX}h$ lies in the closure of the span of $\{f-\Pi_{\cX}f:f\in\cG\}$, which is $\cK_{\rm L}$, that is, if and only if its
$\cK_{\rm C}$ component vanishes.
(iii) $\cK_{\rm L}=\Pi_{\cX^{\perp}}\cG=(\cG+\cX)\cap\cX^{\perp}$, which is closed as an intersection of closed subspaces.
$\cK_{\rm L}\perp\cH_{\rm C}=\cG^{\perp}$ if and only if $\cK_{\rm L}\subseteq\cG$, that is, if and only if $\Pi_{\cX}f\in\cG$ for every $f\in\cG$.
Write $\cX=\mathrm{span}\{1,\eta_0\}\oplus\cX'$; for $f\in\cG$, $\Pi_{\cX}f=\Pi_{\{1,\eta_0\}}f+\Pi_{\cX'}f$ with the first term in $\cG$,
so the condition is $\Pi_{\cX'}f\in\cG$, and by the assumption $\cX'\cap\cG=\{0\}$ this means $\Pi_{\cX'}f=0$ for all $f\in\cG$,
i.e.\ $\cX'\perp\cG$, i.e.\ $\E_w(x'\mid\eta_0)=0$ for every $x'\in\cX'$, with $\E_w$ the $w_0$-weighted conditional expectation.
Because $w_0>0$ is a function of $\eta_0$, $\E_w(x'\mid\eta_0)=\E(w_0x'\mid\eta_0)/\E(w_0\mid\eta_0)=\E(x'\mid\eta_0)$. If
$\E(\tx\mid\eta_0)$ is affine for every column, then for $x'\in\cX'$, $\E(x'\mid\eta_0)=c+d\eta_0$ with
$\ip{c+d\eta_0}{1}=\ip{x'}{1}=0$ and $\ip{c+d\eta_0}{\eta_0}=\ip{x'}{\eta_0}=0$, so $c+d\eta_0=0$ when $\eta_0$ is non-degenerate;
conversely $\E(x'\mid\eta_0)=0$ for all $x'\in\cX'$ makes $\E(\tx\mid\eta_0)$ affine. Under the condition $\cK_{\rm L}\subseteq\cG$ and
$\cK_{\rm L}\perp\mathrm{span}\{1,\eta_0\}$, so $\cK_{\rm L}\subseteq\cH_{\rm L}$; and $\cH_{\rm L}\perp\cX'$ gives $\cH_{\rm L}\subseteq\cX^{\perp}$,
so $\cH_{\rm L}=\Pi_{\cX^{\perp}}\cH_{\rm L}\subseteq\cK_{\rm L}$. For jointly elliptical columns, $\E(\tx\mid a^{\top}\tx)$ is affine in
$a^{\top}\tx$.
\end{proof}

\noindent Proposition~\ref{prop:alias} gives the in-sample verdicts their reading. COV named: no function of the
score repairs the model, so it must be revised. LINK named and COV not named: the data are consistent with a
first-order repair by a change of link or by a smooth function of the score added to the model (Proposition
\ref{prop:alias}(ii)); this is not evidence that no revision is needed, because a non-rejection is not proof and the COV
bases are finite. When the model matrix contains binary covariates, interactions or nonlinear terms, part of a
covariate-structure misfit can appear in LINK, and a covariate-structure misfit may remain after such a repair.

\begin{proposition}[A de-aliased in-sample link group]\label{prop:dealias}
Let $\mathcal E=\mathrm{span}\{\E(x'\mid\eta_0):x'\in\cX'\}\subset\cG$, the part of the score-space that the covariates
carry, and for a link basis $b\in\cG$ let $b^{\rm da}=b-\Pi_{\mathrm{span}\{1,\eta_0\}+\mathcal E}\,b$. Then
\begin{enumerate}[label=(\roman*),leftmargin=*,itemsep=2pt]
\item $b^{\rm da}\in\cG\cap\cX^{\perp}\subseteq\cK_{\rm L}$, so in-sample its statistic needs no further projection, and
  $\ip{b^{\rm da}}{h}=\ip{b^{\rm da}}{h_{\rm L}}$ for every $h$: \emph{the de-aliased in-sample link hypothesis is a statement about
  the external link piece $h_{\rm L}$ alone};
\item since each $x'\in\cX'$ is orthogonal to $1$ and $\eta_0$, so is $\E(x'\mid\eta_0)$, and $\mathcal E\subseteq\cH_{\rm L}$. The price is the
  link directions in $\mathcal E$ (at most $d-1$ of them): a link misfit lying there is partly erased by the refit, and what
  is left of it is seen by neither the de-aliased LINK group nor COV, so it is undetectable in-sample. The price is
  zero exactly when $\mathcal E=\{0\}$, that is, when $\E(X\mid\eta_0)$ is affine; then $b^{\rm da}=\Pi_{\cH_{\rm L}}b$.
\end{enumerate}
\end{proposition}
\begin{proof}
(i) $b^{\rm da}$ is $b\in\cG$ minus an element of $\cG$, so it lies in $\cG$, and it is orthogonal to $1$, $\eta_0$ and
$\mathcal E$. For $x\in\cX$ write $x=c+d\eta_0+x'$ with $x'\in\cX'$; since $b^{\rm da}\in\cG$,
$\ip{b^{\rm da}}{x}=\ip{b^{\rm da}}{\Pi_{\cG}x}$, and $\Pi_{\cG}x=c+d\eta_0+\E_w(x'\mid\eta_0)=c+d\eta_0+\E(x'\mid\eta_0)$ lies in
$\mathrm{span}\{1,\eta_0\}+\mathcal E$, so the inner product is zero. Hence $b^{\rm da}\in\cX^{\perp}$. For any $h$,
$\ip{b^{\rm da}}{h}=\ip{b^{\rm da}}{\Pi_{\cG}h}=\ip{b^{\rm da}}{h_{\rm I}+h_{\rm S}+h_{\rm L}}=\ip{b^{\rm da}}{h_{\rm L}}$.
(ii) $\ip{\E(x'\mid\eta_0)}{1}=\ip{x'}{1}=0$ and likewise for $\eta_0$, so $\mathcal E\subseteq\cH_{\rm L}$. The de-aliased bases span
$\cH_{\rm L}\ominus\mathcal E$; $\dim\mathcal E\le\dim\cX'=d-1$. A link misfit $h_{\rm L}\in\mathcal E$ has in-sample remainder
$h_{\rm L}-\Pi_{\cX}h_{\rm L}$, which is orthogonal to the de-aliased bases (they are in $\cX^{\perp}$ and orthogonal to
$\mathcal E$) and lies in $\cG+\cX$, hence has no $\cK_{\rm C}$ component. $\mathcal E=\{0\}$ if and only if every
$\E(x'\mid\eta_0)$ vanishes, i.e.\ $\E(X\mid\eta_0)$ is affine (Proposition~\ref{prop:alias}iii).
\end{proof}

\noindent The software (option \texttt{dealias}) uses an estimate of $\mathcal E$ with two safeguards. First, a model column is
de-aliased only if its weighted mean given the fitted score is significantly curved (an $F$ test of a
three-degree-of-freedom natural spline in $\hat\eta$ against a line, level $0.05$). Both the choice and the fitted
directions are frozen: the spline fits of the selected columns are computed once, on the observed score, and the same
vectors are projected off in the data and in every bootstrap draw. They are then fixed functions of the covariates
and Lemma~\ref{lem:boot} applies. Without the freeze, when $\E(X\mid\eta_0)$ is affine the estimated curvature is sampling
noise whose direction changes between the data and each bootstrap draw, and the lemma would not apply. Second, for each selected column $\mathcal E$ is replaced by $\mathcal E_3$, the weighted fit of the column on the
spline; the residual aliasing is $\ip{\Pi_{\cX+\mathcal E_3}b-\Pi_{\cX+\mathcal E}b}{h_{\rm C}}$, small when the conditional means are
smooth. The guarantees of Theorem~\ref{thm:fwer} then apply conditionally on the selected columns. The price of
part (ii) can be large: in the design of Table~\ref{tab:dealias}, where the covariates carry most of the curvature
of the score, de-aliasing halved the leak of a covariate misfit into LINK but removed almost all power against a
curved link. The option is therefore best used as a sensitivity check of an in-sample LINK verdict.

\begin{corollary}[Grouping by fitted risk]\label{cor:hl}
Let a statistic be a continuous function of the residual sums $n^{-1/2}\sum_i\mathbf 1\{\hat\eta_i\in A_g\}(Y_i-\hat p_i)$ over a
fixed number of intervals $A_g$ whose end points converge to quantiles of $\eta_0$ at which $\eta_0$ has a continuous
distribution, and of consistently estimated scaling constants (the Hosmer--Lemeshow construction), and let its
null limit law be continuous. Its local power equals its size against any $h\in\cH_{\rm C}$ in external validation (with
$p_0$ in place of $\hat p$), and in-sample against any $h$ with $\Pi_{\cG+\cX}h\in\cX$, in particular against
$h\in\cK_{\rm C}$.
\end{corollary}
\begin{proof}
Each indicator $\mathbf 1_{A_g}=\mathbf 1\{\eta_0\in A_g\}$ belongs to $\cG$. Externally the mean shift of the $g$th sum is
$\ip{\mathbf 1_{A_g}}{h}=0$ for $h\perp\cG$. In-sample, by Theorem~\ref{thm:insample}(a, Step 3) the mean shift is
$\ip{\mathbf 1_{A_g}-\Pi_{\cX}\mathbf 1_{A_g}}{h}$; this vanishes when $\Pi_{\cG+\cX}h\in\cX$, because
$\mathbf 1_{A_g}-\Pi_{\cX}\mathbf 1_{A_g}\in\cG+\cX$ and is orthogonal to $\cX$. The covariances do not depend on $h$; estimated end
points converge, the indicators of intervals form a Donsker class and the scaling constants converge, so by Slutsky's
lemma the limit law of the statistic is its null limit law.
\end{proof}

\noindent This is why the partition tests of the Hosmer--Lemeshow school, and every other test indexed by the
fitted risk alone, detect a curved link but not a covariate-structure misfit uncorrelated with the score (exactly in
external validation; in-sample when $\Pi_{\cG+\cX}h\in\cX$, which an aliased covariate misfit violates), an
observation made repeatedly in simulations (Hosmer et al., 1997; Liu et al., 2024).

\section{The law of the group statistics}

\begin{theorem}[Noncentrality decomposition]\label{thm:law}
Under the local alternatives \eqref{eq:local} and Assumption~\ref{ass:basis}, in external validation and conditionally on
$X_1,X_2,\dots$ for almost every sequence,
\[
(U_{kj})_{k,j}\dto N(\mu,\Sigma),\qquad \mu_{kj}=\ip{\tilde b_{kj}}{h_k},\qquad
\Sigma_{(kj),(k'j')}=\E\{w_0\tilde b_{kj}\tilde b_{k'j'}^{\top}\},
\]
with $\Sigma_{(kj),(k'j')}=0$ whenever $k\neq k'$. Consequently $Q_{kj}\dto\chi^2_{r_{kj}}(\lambda_{kj})$ with
$\lambda_{kj}=\nrm{\Pi_{\tilde b_{kj}}h_k}^2$, where $\Pi_{\tilde b_{kj}}$ projects onto the span of the components of
$\tilde b_{kj}$; statistics of different groups are asymptotically independent; and, taking one member per group,
$\sum_k\lambda_{kj_k}\le\nrm{h}^2$. The Monte Carlo draws have the same limit with $\mu=0$. In-sample, the same holds
unconditionally for the L and C groups with $h$ replaced by $h-\Pi_{\cX}h$ and the pieces by $\cK_{\rm L},\cK_{\rm C}$
(Theorem~\ref{thm:insample}a).
\end{theorem}

\begin{proof}
Conditionally on the covariates, $U$ is a normalized sum of independent mean-zero vectors plus the mean
$n^{-1/2}\sum_i\tilde b(X_i)\{\pi_n(X_i)-p_0(X_i)\}$, which converges almost surely to $\ip{\tilde b_{kj}}{h}$ by
\eqref{eq:rem}, Assumption~\ref{ass:basis} and the strong law; this equals $\ip{\tilde b_{kj}}{h_k}$ because
$\tilde b_{kj}\in\cH_k$. Covariance and Lyapunov condition are as in Step~5 of the proof of Theorem~\ref{thm:insample}.
Cross-group covariances vanish because $\cH_k\perp\cH_{k'}$, and jointly normal uncorrelated blocks are independent.
$\hat V_{kj}\to V_{kj}$ almost surely, so $Q_{kj}\dto Z^{\top}V_{kj}^{-1}Z$ with $Z\sim N(\mu_{kj},V_{kj})$, a noncentral
$\chi^2_{r_{kj}}$ with $\lambda=\mu^{\top}V^{-1}\mu$, the squared norm of the $L^2(w)$ regression of $h_k$ on the components
of $\tilde b_{kj}$. Bessel's inequality gives the bound, because the $\Pi_{\tilde b_{kj_k}}h_k$ lie in mutually orthogonal
pieces. The Monte Carlo draws are the case $\pi_n=p_0$.
\end{proof}

\section{Strong familywise error control}

\begin{lemma}[Rank tests with an asymptotically exchangeable reference]\label{lem:rank}
Let $T_n$ be a statistic and, conditionally on the data, let $T_n^{(1)},\dots,T_n^{(M)}$ be independent draws from a law
$L_n^*$. Suppose $T_n\dto L$ and $d_{\rm BL}(L_n^*,L)\to0$ in probability, with $L$ continuous. Then
$\Pr[\{1+\#(m:T_n^{(m)}\ge T_n)\}/(M+1)\le\alpha]\to\lfloor\alpha(M+1)\rfloor/(M+1)\le\alpha$.
\end{lemma}
\begin{proof}
For bounded Lipschitz $f_0,\dots,f_M$,
$\E\{f_0(T_n)\prod_mf_m(T_n^{(m)})\}=\E\{f_0(T_n)\prod_m\int f_m\,dL_n^*\}$, and $\int f_m\,dL_n^*\to\int f_m\,dL$ in
probability and boundedly; hence $(T_n,T_n^{(1)},\dots,T_n^{(M)})\dto L^{\otimes(M+1)}$. The rank of the first coordinate is
continuous off the set of ties, which has probability zero under the continuous product law, and under
$L^{\otimes(M+1)}$ it is uniform on $\{1,\dots,M+1\}$.
\end{proof}

\begin{lemma}[Continuity of the limit law]\label{lem:cont}
Let $Z\sim N(\mu,\Sigma)$ collect the limits of the members in a set $S$, with each block $V_{kj}$ positive definite, and
let $T(Z)$ be the Cauchy combination \eqref{eq:TS} of the member $p$-values $1-F_{r}(Z_{kj}^{\top}V_{kj}^{-1}Z_{kj})$. Then the
law of $T(Z)$ has no atoms.
\end{lemma}
\begin{proof}
The mean lies in the range of the covariance: $\mu=\E\{w_0\tilde bh\}=\Sigma a$, with $a$ the coefficients of the
$L^2(w)$ regression of $h$ on the members' components. So $Z$ lives on the linear subspace $\mathrm{range}(\Sigma)$, where it
has a positive density with respect to Lebesgue measure; the joint covariance may be singular because members of
one group share directions. Each block $Z_{kj}$
is a non-degenerate linear image of $Z$ on this set, so $\{Z_{kj}=0\}$ is a proper linear subspace, of measure zero.
Off these sets every member $p$-value lies in $(0,1)$ and $T$ is real-analytic. The complement of the finitely many
subspaces is open with finitely many connected components (two for a single direction, when $r=1$). On each component
$T$ is non-constant: along a ray $z\mapsto tz$, $t\to\infty$, every quadratic form tends to infinity and $T\to+\infty$. The
level sets of a non-constant real-analytic function on a connected open set have measure zero, so $\Pr\{T(Z)=c\}=0$ for
every $c$.
\end{proof}

\begin{lemma}[The parametric bootstrap under local alternatives]\label{lem:boot}
Under Assumptions~\ref{ass:basis}--\ref{ass:fit}, let the $L$ and $C$ members be computed on bootstrap data
$Y^*_i\sim\mathrm{Bernoulli}(\hat p_i)$ with the refitted $\hat\beta^*$. Then, in probability, their conditional law given the data
converges to $N(0,\Sigma^{\rm in})$, the in-sample limit of Theorem~\ref{thm:law} with $h=0$.
\end{lemma}
\begin{proof}
The bootstrap world is the logistic model at $\hat\beta$, with $\hat\beta\to\beta_0$ in probability. Steps 1--5 of the proof of
Theorem~\ref{thm:insample} hold conditionally on the data, with $h=0$ and $\beta_0$ replaced by $\hat\beta$, provided the
expansions are uniform for $\beta$ in a neighborhood of $\beta_0$. They are: the maximum-likelihood expansion holds
uniformly on compact parameter sets for the logistic model with $I_0$ positive definite; the averages
$n^{-1}\sum_iw_\beta\tilde b\tx^{\top}$ and $n^{-1}\sum_i\partial_\theta b$ obey a uniform law of large numbers on the
neighborhood under the envelope conditions of Assumptions~\ref{ass:basis}--\ref{ass:fit}; and the conditional
Lyapunov condition holds with the same moments. This is the argument of Stute et al.\ (1998) for the parametric
bootstrap of model checks, here for finitely many directions. The software discards bootstrap draws whose refit fails;
under the assumptions the probability of a failure tends to zero, so this does not change the limit.
\end{proof}

\begin{theorem}[Strong familywise error control]\label{thm:fwer}
Let $\mathcal T=\{k:H_k\text{ true}\}$ (external validation) or $\{k:H^{\rm in}_k\text{ true}\}$ (in-sample). Under the local
alternatives \eqref{eq:local} and Assumption~\ref{ass:basis},
\[
\limsup_{n\to\infty}\Pr(\text{some }k\in\mathcal T\text{ is named})\le\frac{\lfloor\alpha(M+1)\rfloor}{M+1}\le\alpha
\]
in external validation, whatever the other parts of $h$ are. Under Assumption~\ref{ass:fit} as well, the same bound with
$B$ in place of $M$ holds in-sample.
\end{theorem}

\begin{proof}
If $\mathcal T$ is empty there is nothing to prove. Otherwise a group $k\in\mathcal T$ is named only if $p_S\le\alpha$ for every
$S\ni k$, in particular for $S=\mathcal T$, so the event is contained in $\{p_{\mathcal T}\le\alpha\}$ (the closure principle).
$T_{\mathcal T}$ uses only members of groups in $\mathcal T$. By Theorem~\ref{thm:law} their joint limit is
$N(0,\Sigma_{\mathcal T})$, the same as for the reference draws, because their means vanish by the definition of the
hypotheses and the covariance does not depend on $h$. The map from the members' statistics to $T_{\mathcal T}$ is continuous
off the null set where some quadratic form is zero, so $T_{\mathcal T}$ and its reference copies share a limit law $L$,
continuous by Lemma~\ref{lem:cont}. Externally, the reference draws are exactly the null at sample size $n$; in-sample,
Lemma~\ref{lem:boot} gives their limit. Lemma~\ref{lem:rank} completes the proof.
\end{proof}

\begin{remark}
The proof uses neither the orthogonality of the pieces nor any property of the Cauchy combination: \emph{any}
statistic calibrated by these references gives a valid closed procedure, and the Cauchy combination is chosen only
for power. Orthogonality is what gives the hypotheses their meaning: externally $H_k$ is a statement about the $k$th piece of
the misfit and nothing else; in-sample, see Proposition~\ref{prop:alias}(iii).
\end{remark}

\section{Consistent localization and the right rung of the ladder}

\begin{theorem}[Consistent localization]\label{thm:consist}
Assume $\alpha\ge1/(M+1)$ (or $\alpha\ge1/(B+1)$ in-sample). The statements are written for external validation; in-sample
read $\Pi_{\cK_k}b_{kj}$ for $\tilde b_{kj}$, $H^{\rm in}_k$ for $H_k$, and, for fixed departures, the hypothesis
$\E\{\Pi^{*}_{\cK_k}b_{kj}(X)(\pi(X)-p^*(X))\}=0$, with the projections taken in $L^2(w^*)$ at the pseudo-true fit.
\begin{enumerate}[label=(\alph*),leftmargin=*,itemsep=2pt]
\item (Growing local departures.) Let $h=c\,h^{0}$ with $h^0$ fixed. If $\ip{\tilde b_{kj}}{h^0_k}\neq0$ for some member $j$ of
  group $k$, then $\lim_{c\to\infty}\liminf_{n\to\infty}\Pr(k\text{ named})=1$. If $H_k$ holds, then
  $\limsup_{n}\Pr(k\text{ named})\le\alpha$ for every $c$.
\item (Fixed departures.) Let $\pi$ be fixed and suppose $\ip{\tilde b_{kj}}{g}\neq0$ for some $j$ (external validation). Then
  $\Pr(k\text{ named})\to1$ as $n\to\infty$, up to the numerical truncation of Remark~\ref{rem:clip}. In-sample, the same holds
  with $p_0$, $w_0$ and the projections replaced by those of the pseudo-true fit of Theorem~\ref{thm:insample}(c).
\end{enumerate}
\end{theorem}

\begin{proof}
(a) The second statement is Theorem~\ref{thm:fwer}. For the first, the limit of Theorem~\ref{thm:law} gives member $(k,j)$ the
law $\chi^2_r(c^2\lambda^0)$ with $\lambda^0>0$, so as $c\to\infty$, $P_{kj}\to0$ in probability and its Cauchy coordinate
$\tan\{(\frac12-P_{kj})\pi\}\sim1/(\pi P_{kj})\to+\infty$. Every other member of a set $S\ni k$ has either a fixed limit
law, so a tight coordinate, or a noncentrality growing like $c^2$, so a coordinate tending to $+\infty$; no coordinate
tends to $-\infty$. Hence $T_S\to+\infty$ in probability for every $S\ni k$, while the reference law of $T_S$ is the fixed
null law, continuous by Lemma~\ref{lem:cont}. So $\Pr(T_S>\max_mT_S^{(m)})\to1$ and $p_S=1/(M+1)\le\alpha$ for each of
the finitely many $S\ni k$.

(b) $n^{-1/2}U_{kj}\to\ip{\tilde b_{kj}}{g}\neq0$ almost surely, so $Q_{kj}\ge\kappa n$ for some $\kappa>0$ with probability tending
to one. Members with $\ip{\tilde b}{g}=0$ have $U=O_p(1)$ (bounded variance), hence tight coordinates; members with non-zero
mean diverge to $+\infty$. The reference law converges and is tight. The rest is as in (a). In-sample,
$\hat\beta\to\beta^*$ and the argument applies around the pseudo-true fit.
\end{proof}

\begin{corollary}[The decision rule]\label{cor:rung}
Order the groups I $<$ S $<$ L $<$ C and let $k^*$ be the highest group with a detectable component,
$\ip{\tilde b_{k^*j}}{h^0_{k^*}}\ne0$ for some $j$. Under local alternatives $h=c\,h^0$,
\[
\lim_{c\to\infty}\liminf_{n\to\infty}\Pr(\text{highest named group}=k^*)\ge1-\alpha .
\]
\end{corollary}
\begin{proof}
The highest named group equals $k^*$ unless $k^*$ is not named, which has limit probability zero by
Theorem~\ref{thm:consist}(a), or some $k>k^*$ is named; all such $k$ are true hypotheses, and by Theorem~\ref{thm:fwer}
the probability that any of them is named is at most $\alpha$ in the limit.
\end{proof}

\noindent The guarantee is local. For fixed departures the higher null groups can be named more often than $\alpha$
(Proposition~\ref{prop:beyond}(iii)), and a component orthogonal to every basis of its group is not detectable at all.

\section{Other prediction models, model updating, planning and optimality}

\begin{corollary}[Any prediction model]\label{cor:anymodel}
In external validation the predictions $p_0(x)$ may come from any model fitted on other data (a neural network, a
random forest, a points score, a published equation), and $X$ may be any covariates recorded in the validation
sample, whether or not the model used them. Let $W=(X,\eta_0)$ with $\eta_0=\mathrm{logit}\,p_0$ (the score may depend on inputs
not in $X$), let $\pi=\Pr(Y=1\mid W)$, and build the pieces of Definition~\ref{def:pieces} in $L^2(w)$ of functions of $W$,
with $\cG$ the functions of $\eta_0$ and $\cX=\mathrm{span}\{1,X\}$. Then Theorems~\ref{thm:law}--\ref{thm:consist},
Corollary~\ref{cor:rung} and Proposition~\ref{prop:beyond} hold with the same proofs, the null being $\E(Y\mid W)=p_0$, and the
Monte Carlo reference is exact under that null for any sample size.
\end{corollary}
\begin{proof}
The external-validation proofs use only that $p_0$ is fixed given the conditioning variables, with values in $(0,1)$
and independent of the validation outcomes; they never use that $\eta_0$ is linear in $x$ or how $p_0$ was obtained.
The pieces $\cH_{\rm I},\cH_{\rm S},\cH_{\rm L}$ are defined through $\eta_0$ alone and $\cH_{\rm C}=\cG^{\perp}$ needs no model.
Under the null, $Y_i\sim\mathrm{Bernoulli}\{p_0(W_i)\}$ independently given the $W_i$, exactly the law of the Monte Carlo
draws, so the observed statistic and the $M$ draws are exchangeable (Barnard's argument).
\end{proof}

\noindent For a machine-learning model the four verdicts keep their meaning: INTERCEPT and SLOPE call for the
logistic recalibration of the model's output, LINK for a flexible recalibration of it, and COV for retraining or
adding information, since no recalibration of the output can repair misfit within covariate patterns. When COV is
named through covariates the model never used, the verdict means that the model lacks information, not that it is
mis-specified in the usual sense. Two practical caveats: for a score with few distinct values (a points score), $\cG$ is
spanned by the indicators of its levels, and the spline projections of the software should be replaced by them; and
predictions evaluated on the data the model was trained on are not covered, because the in-sample theory needs the
model's tangent space (Theorem~\ref{thm:insample}d), which a flexible learner does not have in a usable form.

\begin{proposition}[Updating at a rung removes the parts up to that rung]\label{prop:update}
In external validation under the local alternatives \eqref{eq:local}, update the frozen model on the validation
data by maximum likelihood in one of the families
\[
\mathcal U_{\rm I}:\ \eta_0+a,\qquad \mathcal U_{\rm S}:\ a+(1+b)\eta_0,\qquad
\mathcal U_{\rm L}:\ a+(1+b)\eta_0+\gamma^{\top}f(\eta_0),
\]
with $f$ a fixed spline basis in $\eta_0$ of dimension $K$. Let $\mathcal T_{\rm I}=\cH_{\rm I}$, $\mathcal T_{\rm S}=\cH_{\rm I}\oplus\cH_{\rm S}$ and
$\mathcal T_{\rm L}^K=\mathrm{span}\{1,\eta_0,f(\eta_0)\}$, with tangent directions $t_k$ and $I_k=\E\{w_0t_kt_k^{\top}\}$. Then
\[
\sqrt n(\hat\eta_k-\eta_0)\dto\Pi_{\mathcal T_k}h+t_k^{\top}Z_k,\qquad Z_k\sim N(0,I_k^{-1}),
\]
so the misfit left after updating, $h-\sqrt n(\hat\eta_k-\eta_0)$, has, in the limit, the component
\[
h-\Pi_{\mathcal T_k}h=\sum_{k'>k}h_{k'}\quad(k={\rm I},{\rm S}),\qquad h-\Pi_{\mathcal T^K_{\rm L}}h=h_{\rm C}+(h_{\rm L}-\Pi_{\mathcal T_{\rm L}^K}h_{\rm L})
\]
outside $\mathcal T_k$, and a mean-zero estimation error $-t_k^{\top}Z_k$ inside $\mathcal T_k$ with $\E\nrm{t_k^{\top}Z_k}^2=\dim\mathcal T_k$. The
bracket tends to zero as $K\to\infty$ when $\bigcup_K\mathcal T^K_{\rm L}$ is dense in $\cG$ (for example natural splines whose
boundary knots move out to the $1/K$ and $1-1/K$ quantiles of $\eta_0$). Updating at the highest rung with a non-zero part
therefore removes the systematic misfit, up to estimation noise of the usual order; a lighter update leaves at least
that part; and no recalibration removes $h_{\rm C}$.
\end{proposition}
\begin{proof}
Each family is a regular parametric model with offset $\eta_0$ whose tangent space at the null parameter is
$\mathcal T_k$. By Le Cam's third lemma, as in Step~1 of the proof of Theorem~\ref{thm:insample}, the estimate satisfies
$\sqrt n\,\hat\theta_k\dto N(I_k^{-1}\E\{w_0t_kh\},I_k^{-1})$, so $\sqrt n(\hat\eta_k-\eta_0)=t_k^{\top}\sqrt n\hat\theta_k+o_p(1)$ converges to
$t_k^{\top}I_k^{-1}\E\{w_0t_kh\}+t_k^{\top}Z_k$; the first term is the $L^2(w)$ projection $\Pi_{\mathcal T_k}h$, and
$\E\nrm{t_k^{\top}Z_k}^2=\mathrm{tr}(I_kI_k^{-1})=\dim\mathcal T_k$. The display follows from Definition~\ref{def:pieces}, because $\mathcal T_{\rm I}$ and
$\mathcal T_{\rm S}$ are sums of the first pieces and $\mathcal T^K_{\rm L}\subseteq\cG$. If the spline dimension grows with $n$, the
estimation error has squared norm of order $(K+2)/n$, which exceeds the $n^{-1}$ scale of the misfit, so ``removes $h_{\rm L}$''
then holds in the bias sense only.
\end{proof}

\noindent Together with Corollary~\ref{cor:rung}, Proposition~\ref{prop:update} turns the decision table of the paper
into a statement: the highest named group points to the lightest update that removes the detectable systematic
misfit. The updated model must then be checked on new data, because on the updating data Lemma~\ref{lem:erase} and
Theorem~\ref{thm:insample} erase exactly the parts it fitted.

\begin{corollary}[Power and sample size]\label{cor:power}
In external validation, member $(k,j)$ has first-order power
$\mathrm{pow}_r(\lambda)=1-F_{\chi^2_r(\lambda)}\{F_r^{-1}(1-\alpha)\}$ with $r=r_{kj}$ and
$\lambda=n\,\mu_{kj}^{\top}V_{kj}^{-1}\mu_{kj}$, $\mu_{kj}=\ip{\tilde b_{kj}}{g}$ (Theorem~\ref{thm:law}, with $g$ in place of
$n^{-1/2}h$). For the one-column members this gives closed forms. Let $\lambda^*$ be the noncentrality with
$\mathrm{pow}_1(\lambda^*)$ equal to the target power ($\lambda^*=(z_{0.975}+z_{0.80})^2=7.85$ for $\alpha=0.05$ and power $0.80$).
\begin{enumerate}[label=(\alph*),leftmargin=*,itemsep=2pt]
\item calibration in the large: a mean miscalibration $\delta=\E\{\pi(X)-p_0(X)\}$ is detected with the target power when
  $n\ge\lambda^*\E(w_0)/\delta^2$;
\item calibration slope: a slope $1+\epsilon$ on the logit scale, $\mathrm{logit}\,\pi=a+(1+\epsilon)\eta_0$ with $a=O(\epsilon)$, gives
  $\mu=\epsilon\,\E\{w_0(\eta_0-\bar\eta)^2\}+O(\epsilon^2)$, so $n\ge\lambda^*/[\epsilon^2\E\{w_0(\eta_0-\bar\eta)^2\}]$.
\end{enumerate}
Both expressions involve only the covariates of the validation sample and the frozen predictions, so they can be
computed before the outcomes are collected.
\end{corollary}
\begin{proof}
Theorem~\ref{thm:law} gives the noncentral $\chi^2$ limit; replacing $n^{-1/2}h$ by $g$ is the usual first-order
approximation for a fixed alternative, which also ignores the change of variance (Proposition~\ref{prop:beyond}iii). For
(a) the basis is $1$, so $\mu=\E(\pi-p_0)=\delta$ and $V=\E(w_0)$. For (b) the basis is $\eta_0-\bar\eta$ and, by \eqref{eq:rem},
$\pi-p_0=w_0\{a'+\epsilon\eta_0\}+O(\epsilon^2)$ for some $a'=O(\epsilon)$; the constant is orthogonal to $\eta_0-\bar\eta$ because $\bar\eta$ is the
$w_0$-weighted mean, so $\mu=\epsilon\E\{w_0\eta_0(\eta_0-\bar\eta)\}=\epsilon\E\{w_0(\eta_0-\bar\eta)^2\}$ and $V=\E\{w_0(\eta_0-\bar\eta)^2\}$.
\end{proof}

\noindent The member's power is a benchmark for the group, not a bound: the group combines several members, and the
closure asks every intersection containing it to reject. Part (a) is the usual power calculation for an
observed-versus-expected test; the expressions complement the sample-size rules for external validation studies of
Riley et al.\ (2021), which target the precision of calibration estimates rather than the power to locate a misfit.

\begin{proposition}[Each member is locally optimal in its own directions]\label{prop:optimal}
Fix a member with basis $\tilde b$ of $r$ columns and consider the local model
$\mathrm{logit}\,\pi_n=\eta_0+n^{-1/2}\gamma^{\top}\tilde b$, $\gamma\in\mathbb R^r$, in external validation (the member's own local model,
with no nuisance parameter). Among all tests of $\gamma=0$ with asymptotic level $\alpha$, the member's test maximizes the
minimum asymptotic power over every ellipsoid $\{\gamma:\gamma^{\top}V\gamma=c\}$ (it is asymptotically maximin and most
stringent); for $r=1$ it is asymptotically uniformly most powerful among asymptotically locally unbiased tests (those
whose limiting power is at least $\alpha$ for every local $\gamma$). Against a general local departure $h$ its power depends on
$h$ only through $\Pi_{\tilde b}h$.
\end{proposition}
\begin{proof}
The local model is locally asymptotically normal with Fisher information $V=\E\{w_0\tilde b\tilde b^{\top}\}$ and central
sequence $U$, so it converges to the Gaussian shift experiment $Z\sim N(V\gamma,V)$. In that experiment the test that
rejects for large $Z^{\top}V^{-1}Z$ is maximin over the ellipsoids and most stringent (Wald, 1943), and for $r=1$ the
two-sided $|Z|$ test is uniformly most powerful unbiased; these properties transfer to the sequence of experiments,
and the member's $Q$ is asymptotically equivalent to $U^{\top}V^{-1}U$ (Lehmann and Romano, 2005, Section~13.5; van der
Vaart, 1998, Chapter~15). The last statement is Theorem~\ref{thm:law}.
\end{proof}

\noindent No optimality is claimed for the groups: a Cauchy combination of several members gives up some power in
any one direction to keep power in all of them, which is what a localization procedure needs.

\section{Outside the local regime}

For a fixed departure the logit and probability scales no longer agree and the variance changes. Proposition~\ref{prop:beyond}
separates the effects; Section~\ref{sec:numerics} measures them on the simulation designs.

\begin{proposition}\label{prop:beyond}
Let $\pi$ be fixed, external validation.
\begin{enumerate}[label=(\roman*),leftmargin=*,itemsep=2pt]
\item (Purity is a property of the probability scale.) The mean of $n^{-1/2}U_{kj}$ is $\ip{\tilde b_{kj}}{g}$ exactly, at
  every $n$. A departure that is pure on the logit scale, $\mathrm{logit}\,\pi=\eta_0+h$ with $h\in\cH_k$, has
  $g=h+\frac12(1-2p_0)h^2+O(h^3)$, whose second-order term generally has components in other pieces. For example
  $h=s(\eta_0-\bar\eta)\in\cH_{\rm S}$ gives a component of order $s^2$ in $\cH_{\rm L}$, because $(1-2p_0)(\eta_0-\bar\eta)^2$ is not
  an affine function of $\eta_0$. Naming LINK for a large logit-scale slope departure is then a correct detection on the
  probability scale.
\item (Probability-scale-pure departures are bounded.) A departure $g=s\,g^0$ needs $s\,p_0g^0\le1$ and
  $s(1-p_0)g^0\ge-1$ almost surely. For non-zero $g^0$ these confine $s$ to a bounded interval $[-a,b]$, with $b=0$ when
  $p_0(g^0)^+$ or $(1-p_0)(g^0)^-$ is unbounded and $a=0$ when $p_0(g^0)^-$ or $(1-p_0)(g^0)^+$ is unbounded. With Gaussian
  covariates the interval is $\{0\}$ for $\eta_0^2$ and $x_1x_3$, and for $\eta_0-\bar\eta$ only $s\le0$ (risks made less extreme) is
  feasible. The simulated ``probability-scale-pure'' families ($s>0$) must therefore truncate the true risks, and the
  truncation adds components outside $\cH_k$.
\item (The variance changes.) Unconditionally, if $\ip{\tilde b_{kj}}{g}=0$, then $U_{kj}\dto N(0,\Omega_{kj})$ with
\[
\Omega_{kj}=\E\big[\tilde b_{kj}\tilde b_{kj}^{\top}\{\pi(1-\pi)+(\pi-p_0)^2\}\big]=\E\big[\tilde b_{kj}\tilde b_{kj}^{\top}\{w_0+(1-2p_0)(\pi-p_0)\}\big],
\]
  and $Q_{kj}\dto\sum_l\nu_l\chi^2_{1,l}$ with $\nu_l$ the eigenvalues of $V_{kj}^{-1}\Omega_{kj}$. The Monte Carlo reference uses
  $p_0$, so the size of the group is distorted, upwards when all $\nu_l\ge1$ with one strictly larger. Members of different
  groups become correlated, because orthogonality with weight $w_0$ does not imply orthogonality with weight
  $w_0+(1-2p_0)(\pi-p_0)$. The distortion is first order in $\pi-p_0$.
\end{enumerate}
\end{proposition}

\begin{proof}
(i) is the identity \eqref{eq:g} and the expansion $\expit(\eta_0+h)=p_0+w_0h+\frac12w_0(1-2p_0)h^2+O(h^3)$, divided by $w_0$.
(ii) For $s>0$ the constraints are $s\,p_0(g^0)^+\le1$ and $s(1-p_0)(g^0)^-\le1$; for $s<0$ they are $|s|p_0(g^0)^-\le1$ and
$|s|(1-p_0)(g^0)^+\le1$. A non-zero $g^0$ makes one constraint on each side binding, and an unbounded product forces that
side to zero.
(iii) Unconditionally the summands $\tilde b(X_i)\{Y_i-p_0(X_i)\}$ are independent and identically distributed with mean
$\ip{\tilde b}{g}=0$ and second moment $\E[\tilde b\tilde b^{\top}\E\{(Y-p_0)^2\mid X\}]=\Omega$; the multivariate central limit
theorem and Slutsky's lemma give the limit. Conditionally on the covariates the mean
$n^{-1/2}\sum_i\tilde b(X_i)(\pi-p_0)(X_i)$ does not vanish but is a random $O_p(1)$ quantity, which is why the statement is
unconditional.
\end{proof}

\noindent Part (iii) is the only obstacle to error control beyond the local regime once the hypotheses are stated on
the probability scale, and it can be removed by calibrating with the observed residuals instead of the model.

\paragraph{The multiplier calibration.} The software projects each member's basis $b_{kj}$ in the sample, off a set of
directions $Z_{kj}$ (none for INTERCEPT; $1$ for SLOPE; $1,\eta_0$ for LINK; $1,\eta_0$ and a spline in $\eta_0$ for COV), so that
$\sum_iw_0\tilde b_{kj}Z_{kj}^{\top}=0$ exactly. Let $e_{kj,i}=\{Y_i-p_0(X_i)\}-w_0(X_i)Z_{kj}(X_i)^{\top}\hat\theta_{kj}$ be the residuals after
refitting those directions as an offset, $\hat\theta_{kj}=(\sum_iw_0Z_{kj}Z_{kj}^{\top})^{-1}\sum_iZ_{kj}(Y_i-p_0)$. Each member keeps its
statistic $U_{kj}=n^{-1/2}\sum_i\tilde b_{kj}(Y_i-p_0)=n^{-1/2}\sum_i\tilde b_{kj}e_{kj,i}$, uses the robust variance
$\hat V^{\rm r}_{kj}=n^{-1}\sum_i\tilde b_{kj}\tilde b_{kj}^{\top}e_{kj,i}^2$, and has reference draws
$U^{(m)}_{kj}=n^{-1/2}\sum_i\xi^{(m)}_i\tilde b_{kj}(X_i)e_{kj,i}$, with Rademacher signs $\xi^{(m)}_i$ shared by all members (software:
\texttt{calibration = "multiplier"}). The probability-scale hypotheses are $H^{\rm p}_k:\ \ip{\tilde b_{kj}}{g}=0$ for all $j$,
with $\tilde b_{kj}$ the population projection.

\begin{theorem}[Error control for departures of any size]\label{thm:multiplier}
Let $\pi$ be fixed, let $\Omega'_{kj}=\E[\tilde b_{kj}\tilde b_{kj}^{\top}\{Y-p_0-w_0Z_{kj}^{\top}\theta_{kj}\}^2]$ be positive definite, with
$\theta_{kj}$ the population value of $\hat\theta_{kj}$, and $\E\{(1+\lVert Z_{kj}\rVert^2)\lVert\tilde b_{kj}\rVert^{2+\delta}\}<\infty$. With the multiplier
calibration,
\[
\limsup_{n\to\infty}\Pr(\text{some }k\text{ with }H^{\rm p}_k\text{ true is named})\le\frac{\lfloor\alpha(M+1)\rfloor}{M+1}\le\alpha ,
\]
whatever the size of the departure in the other groups, and Theorem~\ref{thm:consist}(b) continues to hold.
\end{theorem}
\begin{proof}
\emph{The observed statistics.} Because the sample basis is exactly orthogonal to $w_0Z_{kj}$, $U_{kj}$ is unchanged if the
population offset $w_0Z_{kj}^{\top}\theta_{kj}$ is subtracted from the residuals; the difference between the sample and
population projections of the basis multiplies $n^{-1/2}\sum_iZ_{kj}w_0\{g-Z_{kj}^{\top}\theta_{kj}\}$-type sums of order $O_p(1)$ by a
coefficient error of order $O_p(n^{-1/2})$, so $U_{kj}=n^{-1/2}\sum_i\zeta_{kj,i}+o_p(1)$ with
$\zeta_{kj,i}=\tilde b_{kj}(X_i)\{Y_i-p_0(X_i)-w_0Z_{kj}^{\top}\theta_{kj}\}$. For $k\in\mathcal T$, the set of true $H^{\rm p}_k$, the $\zeta_i$ are
independent and identically distributed with mean $\ip{\tilde b_{kj}}{g}-\ip{\tilde b_{kj}}{Z_{kj}^{\top}\theta_{kj}}=0$ (the second
term is zero because $\tilde b_{kj}\perp Z_{kj}$), so $n^{-1/2}\sum_i\zeta_i\dto N(0,\Omega'_{\mathcal T})$ unconditionally, with
$\Omega'_{\mathcal T}=\E\zeta\zeta^{\top}$.

\emph{The reference.} $\hat\theta_{kj}\to\theta_{kj}$, so $e_{kj,i}$ differs from the population residual by
$w_0Z_{kj}^{\top}(\hat\theta_{kj}-\theta_{kj})$ and, under the moment condition, $n^{-1}\sum_i\hat\zeta_i\hat\zeta_i^{\top}\to\Omega'_{\mathcal T}$ almost surely,
where $\hat\zeta_i$ uses $e_i$. Conditionally on the data the reference vector $n^{-1/2}\sum_i\xi_i\hat\zeta_i$ has mean zero and that
covariance, and the Lindeberg condition holds by the $2+\delta$ moment, so its conditional law converges to
$N(0,\Omega'_{\mathcal T})$ almost surely (the multiplier central limit theorem; van der Vaart and Wellner, 1996,
Section~2.9). Because $\xi_i^2=1$, every reference draw has exactly the same robust variance as the data, so using the
one matrix $\hat V^{\rm r}_{kj}\to\Omega'_{kj}$ for the data and all draws is right (it would not be for Gaussian multipliers).
Hence $T_{\mathcal T}$ and its reference copies share a continuous limit law (Lemma~\ref{lem:cont} with $\Sigma=\Omega'_{\mathcal T}$,
$\mu=0$); the observed limit is unconditional and the draws are conditionally independent with bounded-Lipschitz
distance to that law tending to zero almost surely, so Lemma~\ref{lem:rank} and the closure argument of
Theorem~\ref{thm:fwer} conclude. For consistency, a member with $\ip{\tilde b}{g}\neq0$ has $Q\ge\kappa n$ with probability
tending to one, while the reference draws stay tight because the signs remove the mean and the conditional covariance
converges.
\end{proof}

\noindent Without the refit of $Z_{kj}$ the robust variance would estimate $\E[\tilde b\tilde b^{\top}\{Y-p_0\}^2]$ instead of
$\Omega'$; the two differ by a term of indefinite sign whenever $g$ has a component along $Z_{kj}$. For the slope member
under a pure probability-scale intercept departure, $\pi=p_0+0.9w_0$ with $n=2{,}000$, the ratio of the variance of $U$ to
the mean robust variance over 1,500 data sets was 0.905 without the refit and 1.017 with it (standard error about
0.037; deposited script). The two calibrations trade guarantees. The Monte Carlo reference is exact at every sample size when
the model is right, and controls the error asymptotically for local departures. The multiplier reference is only
asymptotic, also under no misfit, but controls the error for departures of any size, stated on the probability
scale. When purity holds exactly in the sample (as in the simulations, where $g$ is orthogonalized on each data set),
each member's robust variance exceeds its conditional variance, so each member's test is conservative. The group
tests were not: in those designs, and under no misfit, the multiplier calibration was liberal at $n\le4{,}000$ (Web
Table~\ref{tab:robust}), as asymptotic calibrations often are in moderate samples.

\section{Implementation}

\begin{remark}[The software's projections]\label{rem:impl}
The software projects with the empirical weights $w_0(X_i)$ (externally) or $\hat w_i$ (in-sample) by weighted least
squares. Externally, conditionally on the covariates, the bases are fixed vectors shared by the data and the Monte
Carlo draws, so the guarantees concern exactly the software's statistics. The software replaces $\cG$ by
$\cG_5=\mathrm{span}\{1,\eta,\mathrm{ns}(\eta,5)\}$ when projecting the covariate bases, so its COV hypothesis is
$\ip{b-\Pi_{\cG_5}b}{h}=0$, not $\Pi_{\cH_{\rm C}}h=0$; Theorem~\ref{thm:fwer} covers the former. The difference is the COV
reading of the link part,
\[
|\ip{(\Pi_{\cG}-\Pi_{\cG_5})b}{h_{\rm L}}|\le\nrm{(\Pi_{\cG}-\Pi_{\cG_5})b}\,\nrm{h_{\rm L}},
\]
which is about 2\% of $\nrm{b-\Pi_{\cG_5}b}$ for the bases of the paper (Table~\ref{tab:approx}). Under local alternatives
this is negligible. Under a fixed link departure the leaked noncentrality grows like $n$, so in very large samples COV can be named
by the approximation error alone; letting the spline dimension grow with $n$ so that the approximation error is
$o(n^{-1/2})$ removes this (Proposition~\ref{prop:sieve}). With one column in $X$, $\cH_{\rm C}=\{0\}$ and the software drops the
COV group. The natural
splines place their boundary knots at the sample range, which has no fixed limit. Externally all statements are
conditional on the observed design, so the software's hypotheses are defined for that design; limit statements hold
along every subsequence on which the (bounded) member correlations converge, and every subsequence has such a
further subsequence. In-sample, Step~4 of the proof of Theorem~\ref{thm:insample} handles the moving knots.
\end{remark}

\begin{proposition}[A COV group that is exactly covariate structure in large samples]\label{prop:sieve}
Let the COV bases be projected off $\cG_K=\mathrm{span}\{1,\eta_0,\mathrm{ns}(\eta_0,K)\}$ with $K=K_n\to\infty$. Assume that $\eta_0$ has
compact support in the population, and that for every COV basis column $b$ and the link part $g_{\rm L}$ of the departure
the approximation errors satisfy $\nrm{(\Pi_{\cG}-\Pi_{\cG_K})b}=O(K^{-s_b})$ and $\nrm{(\Pi_{\cG}-\Pi_{\cG_K})g_{\rm L}}=O(K^{-s_g})$, as for
functions with $s_b$ and $s_g$ square-integrable derivatives (de Boor, 2001; natural-spline boundary constraints cap
the order at about $5/2$ unless the second derivative vanishes at the boundary). If $\sqrt n\,K_n^{-(s_b+s_g)}\to0$,
$K_n^2/n\to0$ and $s_g>1/2$, the mean of the COV statistics under a departure with $g_{\rm C}=0$ is $o(1)$ even for a departure
of fixed size, so with the multiplier calibration (Theorem~\ref{thm:multiplier}, applied along the triangular array of
bases, which is direct because the variances stay bounded) COV keeps its level against any pure link departure. The
software's choice $K_n=\max\{5,\lceil n^{1/3}\rceil\}$ (option \texttt{cov\_df = "auto"}) satisfies these conditions when
$s_b+s_g>3/2$ and $s_g>1/2$.
\end{proposition}
\begin{proof}
With the empirical projection the COV basis is $b-\hat\Pi_Kb$, exactly orthogonal in the sample to $\cG_K$. Since
$g_{\rm L}\in\cG$, the mean of $U$ is $n^{-1/2}\sum_i(b-\hat\Pi_Kb)w_0g_{\rm L}=n^{-1/2}\sum_i(b-\hat\Pi_Kb)w_0(g_{\rm L}-\Pi_Kg_{\rm L})$
up to the sample--population difference. Its population version is
$\sqrt n\,\ip{b-\Pi_Kb}{(\Pi_{\cG}-\Pi_K)g_{\rm L}}=\sqrt n\,\ip{(\Pi_{\cG}-\Pi_K)b}{(\Pi_{\cG}-\Pi_K)g_{\rm L}}$, because
$(\Pi_{\cG}-\Pi_K)g_{\rm L}\in\cG$ and $b-\Pi_{\cG}b\perp\cG$; by Cauchy--Schwarz it is $O(\sqrt nK^{-(s_b+s_g)})=o(1)$. The estimation
error of the projection coefficients is $O_p\{(K/n)^{1/2}\}$ in $L^2(w)$ when $K^2/n\to0$ (the empirical Gram matrix of a spline basis
converges in operator norm), and multiplied by $\sqrt n\,\nrm{(\Pi_{\cG}-\Pi_K)g_{\rm L}}=O(\sqrt nK^{-s_g})$ it contributes
$O_p(K^{1/2-s_g})=o_p(1)$ for $s_g>1/2$. The variance is bounded, so Theorem~\ref{thm:multiplier} applies with the COV
hypothesis equal, in the limit, to $g_{\rm C}=0$.
\end{proof}

\begin{remark}[The robust option]\label{rem:robust}
With \texttt{robust = TRUE} every continuous covariate, every continuous model column and the score are replaced by
their normal scores $\Phi^{-1}\{\mathrm{rank}/(n+1)\}$ before the bases are built. Each basis column is then a fixed
polynomial or spline of quantities bounded by $\Phi^{-1}\{n/(n+1)\}\le\sqrt{2\log n}$, so one record changes each score
statistic $U$ by at most a constant times $(\log n)^{c}/\sqrt n$, with $c$ the degree of the basis, whatever its covariate
values. Externally the transformed bases are functions of the covariates and of the frozen predictions only, so the Monte
Carlo reference stays exact under the null and Theorems~\ref{thm:law}--\ref{thm:consist} hold with the pieces defined on
the normal-score scale ($\cG$ is unchanged, since the normal score of $\eta_0$ is a monotone function of it, but the
\textsc{slope} and link bases are orthogonalized against $\{1,\Phi^{-1}(\mathrm{rank}\,\eta_0)\}$ instead of $\{1,\eta_0\}$). In-sample the
covariate and model-column scores are computed once from the data and kept fixed in every bootstrap draw, while the
score of each refit is transformed again; ranks are not smooth in $\theta$, so Assumption~\ref{ass:fit} does not hold
and Lemma~\ref{lem:boot} is not claimed for this option. Its level was checked by simulation instead (0.047 at
$n=1{,}000$, Table~\ref{tab:robustopt}).
Two consequences for reading a verdict. First, the pieces change meaning with the scale: a misfit that is a pure slope
error on the scale of $\eta_0$ is linear in $\eta_0$ but not in its normal score when $\eta_0$ is far from Gaussian, so
part of it falls in the robust \textsc{link} basis, and the rung that the robust option names can be one step higher
than the one the default names. \textsc{intercept} and \textsc{cov} are not affected in this way, since the first is a
constant and the second is defined by orthogonality to all functions of the score. Second, the bound on one record's
influence is on the test statistics only; the fitted model itself is not robust, which is why part of the false
verdicts under corrupted records remains (Tables~\ref{tab:corrupt} and~\ref{tab:robustopt}).
\end{remark}

\begin{remark}[Numerical truncation]\label{rem:clip}
The software truncates member $p$-values to $[10^{-12},1-10^{-12}]$ before the Cauchy transform. The theorems are stated
for the untruncated statistic. Under the reference law the probability that any of $J$ members in any of the $M+1$
data sets is truncated is at most about $2\times10^{-12}J(M+1)$, and only such an event can change a rank $p$-value; in
Theorem~\ref{thm:consist}(b) the limit is therefore $1-O(10^{-12}JM)$ rather than $1$.
\end{remark}

\begin{remark}[Scope]\label{rem:scope}
Lemma~\ref{lem:erase} is exact in any dimension, and the Monte Carlo reference of external validation is exact under the
global null in any dimension and for predictions from any model. The other results rest on $\sqrt n$ asymptotics with a
fixed number of covariates; when $d/n$ is not small the maximum-likelihood estimate is biased and its variance is
underestimated (Sur and Cand\`es, 2019), and the in-sample bootstrap is no longer justified. In-sample checking as
implemented assumes the unpenalized logistic fit. Other fits erase their own tangent space
(Theorem~\ref{thm:insample}d), which need not be $\cX$: the in-sample groups must then be projected off that space, a
ridge or lasso fit with an unpenalized intercept still satisfies $\sum_i(Y_i-\hat p_i)=0$, and for a bias-reduced (Firth)
fit $\sum_i(Y_i-\hat p_i)=O(1)$, so the intercept and slope statistics degenerate after scaling by $n^{-1/2}$ and must not
be used. Those fits are outside the results proved here; their external validation is covered.
\end{remark}

\section{Numerical checks}\label{sec:numerics}
All tables below are written by script from the stored per-data-set results; the design is the five-covariate
Gaussian design of the paper's simulations, $\eta_0=-0.5+0.5x_1+0.8x_2+0.6x_3+0.4x_4+0.3x_5$, $M=999$, $B=199$,
$\alpha=.05$.

\paragraph{What the fit erases (Table~\ref{tab:t1}).} After a logit fit the two sums of Lemma~\ref{lem:erase} are below
$10^{-5}$ and the in-sample recalibration is $(0,1)$ to numerical precision, whatever the truth. After a probit fit the
sums reach about 4 and the recalibration slope is off by up to .026: the identities belong to the canonical link
(Theorem~\ref{thm:insample}d).

\paragraph{The spline approximation (Table~\ref{tab:approx}).} The five-degree-of-freedom spline in $\eta$ leaves about
2\% of a COV basis column inside $\mathcal G$, so the COV reading of a link departure is about $4\times10^{-4}$ of its
noncentrality (Remark~\ref{rem:impl}).

\paragraph{Fixed departures (Table~\ref{tab:decomp}).} For each fixed-size departure of the simulations we computed,
on the same data sets the procedure analyzed, the true noncentrality of every group. Where a group other than the
departure's own was named often, it carried a real component on the probability scale, created by the logit scale
(Proposition~\ref{prop:beyond}i) or by the truncation of the true risks (Proposition~\ref{prop:beyond}ii): for the
largest covariate departure, SLOPE has a member with power 1.00 and was named in all data sets. These are correct
detections, not familywise errors. What remains is an excess of up to about .13 over the member's own power where
that power is near .05--.15, the variance effect of Proposition~\ref{prop:beyond}(iii).

\paragraph{Local departures (Tables~\ref{tab:localex}--\ref{tab:localmo}).} Under no misfit the familywise rate is
.035--.060 (standard errors .009--.017). For the large departures of the simulations the wrong-group rate falls as $n$
grows when the departure shrinks like $n^{-1/2}$ and rises when its size is held fixed (pure slope: .52, .21, .095 against
.28, .40, .74 at $n=500$, 2,000, 8,000), the contrast between Theorem~\ref{thm:fwer} and Proposition~\ref{prop:beyond}.
These departures are still far from local at $n=8{,}000$, so the moderate departures were followed further: their
wrong-group rate falls from .12--.17 at $n=1{,}000$ to .050--.085 at $n=16{,}000$ and $32{,}000$; pooled over these two sizes
it is .065 (standard error .006, 1,800 data sets), while the own group is still found in 70--100\% of data sets. The
small remaining excess is the second-order term, which vanishes at rate $n^{-1/2}$. In-sample the pure link departure
goes from .435 to .060 by $n=4{,}000$; the large covariate departure, of size 1.0 on the logit scale even at $n=4{,}000$,
falls from .305 to .145--.195 more slowly.

\begin{table}[ht]\centering\small
\caption{In-sample identities after a logit and a probit fit (200 data sets of size 1,000 per truth; largest absolute value over data sets). Sums are $\sum_i(Y_i-\hat p_i)$ and $\sum_i(Y_i-\hat p_i)\hat\eta_i$; $a$ and $b-1$ are the intercept and slope error of the logistic recalibration of $\hat p$ on the same data.}\label{tab:t1}
\begin{tabular}{lrrrrrrrr}\toprule
& \multicolumn{4}{c}{logit fit} & \multicolumn{4}{c}{probit fit}\\ \cmidrule(lr){2-5}\cmidrule(lr){6-9}
truth & sum & sum$\cdot\hat\eta$ & $a$ & $b-1$ & sum & sum$\cdot\hat\eta$ & $a$ & $b-1$\\ \midrule
logistic (null) & 1.9e-06 & 5.1e-06 & 6.7e-09 & 3.5e-08 & 2.27 & 2.74 & 0.011 & 0.019\\
curved link & 3.5e-11 & 9.9e-11 & 1.9e-15 & 8.4e-15 & 3.93 & 2.94 & 0.019 & 0.026\\
interaction & 2.2e-06 & 4.5e-06 & 2.0e-10 & 2.9e-09 & 4.18 & 1.76 & 0.022 & 0.017\\
\bottomrule\end{tabular}\end{table}

\begin{table}[ht]\centering\small
\caption{Remark~\ref{rem:impl}: relative size $\|(\Pi_{\mathcal G}-\Pi_{\mathcal G_5})b\|_w/\|b-\Pi_{\mathcal G_5}b\|_w$ of the part of a COV basis column that the five-degree-of-freedom spline leaves in $\mathcal G$ ($\mathcal G$ approximated by a 40-degree-of-freedom spline; simulation design, $n=10^5$).}\label{tab:approx}
\begin{tabular}{lrrr}\toprule COV member & columns & median & largest\\ \midrule
poly & 10 & 0.020 & 0.038\\
products & 10 & 0.020 & 0.022\\
spline & 15 & 0.021 & 0.026\\
\bottomrule\end{tabular}\end{table}

\begin{table}[ht]\centering\small
\caption{Fixed-size departures of the simulations (external validation, $n=1{,}000$, the same 100 data sets per cell that the procedure analyzed): groups other than the departure's own that were named in at least 10\% of data sets or carry a member with power at least .10. Noncentrality: the largest member noncentrality of the group computed from the true risks; power: that member's power alone at level .05; named: rate at which the procedure named the group.}\label{tab:decomp}
\begin{tabular}{llrlrrr}\toprule scale & departure & size & group & noncentrality & power & named\\ \midrule
probability-scale & slope p & 0.6 & LINK & 1.09 & 0.14 & 0.00\\
probability-scale & slope p & 1.2 & LINK & 15.36 & 0.95 & 0.95\\
probability-scale & slope p & 2.4 & INTERCEPT & 0.74 & 0.14 & 0.00\\
probability-scale & slope p & 2.4 & LINK & 51.04 & 1.00 & 1.00\\
probability-scale & link p & 1 & INTERCEPT & 0.67 & 0.13 & 0.07\\
probability-scale & link p & 1 & SLOPE & 4.08 & 0.52 & 0.36\\
probability-scale & link p & 2 & INTERCEPT & 7.02 & 0.75 & 0.60\\
probability-scale & link p & 2 & SLOPE & 36.38 & 1.00 & 1.00\\
probability-scale & cov p & 1 & SLOPE & 0.63 & 0.12 & 0.04\\
probability-scale & cov p & 2 & INTERCEPT & 1.81 & 0.27 & 0.12\\
probability-scale & cov p & 2 & SLOPE & 8.19 & 0.80 & 0.72\\
probability-scale & cov p & 2 & LINK & 1.10 & 0.14 & 0.10\\
logit-scale & slope pure & 0.6 & LINK & 1.16 & 0.15 & 0.28\\
logit-scale & slope pure & 0.6 & COV & 0.00 & 0.05 & 0.10\\
logit-scale & link pure & 0.25 & SLOPE & 2.70 & 0.37 & 0.24\\
logit-scale & link pure & 0.25 & COV & 0.15 & 0.05 & 0.10\\
logit-scale & link pure & 0.5 & INTERCEPT & 1.69 & 0.25 & 0.17\\
logit-scale & link pure & 0.5 & SLOPE & 38.77 & 1.00 & 1.00\\
logit-scale & cov pure & 1 & INTERCEPT & 0.48 & 0.11 & 0.06\\
logit-scale & cov pure & 1 & SLOPE & 5.00 & 0.60 & 0.36\\
logit-scale & cov pure & 1 & LINK & 0.93 & 0.13 & 0.14\\
logit-scale & cov pure & 2 & INTERCEPT & 3.99 & 0.50 & 0.40\\
logit-scale & cov pure & 2 & SLOPE & 32.17 & 1.00 & 1.00\\
logit-scale & cov pure & 2 & LINK & 9.82 & 0.67 & 0.53\\
\bottomrule\end{tabular}\end{table}

\begin{table}[ht]\centering\small
\caption{External validation, large departures: rate at which every group with a part of the departure is named (own) and at which any other group is named (wrong; nominal bound .05; Monte Carlo standard error at most 0.025). Local track: size $C\sqrt{1000/n}$; fixed track: size $C$. 400 data sets per cell.}\label{tab:localex}
\begin{tabular}{llrrrrrr}\toprule
 & & \multicolumn{2}{c}{$n=500$} & \multicolumn{2}{c}{$n=2{,}000$} & \multicolumn{2}{c}{$n=8{,}000$}\\
\cmidrule(lr){3-4} \cmidrule(lr){5-6} \cmidrule(lr){7-8}
departure & track & own & wrong & own & wrong & own & wrong\\ \midrule
none & local & -- & 0.048 & -- & 0.035 & -- & 0.048\\
slope only & local & 1.00 & 0.520 & 1.00 & 0.205 & 0.99 & 0.095\\
slope only & fixed & 0.95 & 0.278 & 1.00 & 0.400 & 1.00 & 0.735\\
link only & local & 1.00 & 1.000 & 1.00 & 0.985 & 1.00 & 0.405\\
link only & fixed & 1.00 & 0.912 & 1.00 & 1.000 & 1.00 & 1.000\\
covariate only & local & 1.00 & 0.998 & 1.00 & 0.998 & 1.00 & 0.863\\
covariate only & fixed & 1.00 & 0.945 & 1.00 & 1.000 & 1.00 & 1.000\\
link + covariate & local & 1.00 & 1.000 & 1.00 & 1.000 & 1.00 & 1.000\\
link + covariate & fixed & 1.00 & 1.000 & 1.00 & 1.000 & 1.00 & 1.000\\
\bottomrule\end{tabular}\end{table}

\begin{table}[ht]\centering\small
\caption{In-sample checking: rate at which every group with a part of the departure is named (own) and at which any other group is named (wrong; nominal bound .05; Monte Carlo standard error at most 0.035). Local track: size $C\sqrt{1000/n}$; fixed track: size $C$. 200 data sets per cell.}\label{tab:localin}
\begin{tabular}{llrrrrrr}\toprule
 & & \multicolumn{2}{c}{$n=500$} & \multicolumn{2}{c}{$n=2{,}000$} & \multicolumn{2}{c}{$n=4{,}000$}\\
\cmidrule(lr){3-4} \cmidrule(lr){5-6} \cmidrule(lr){7-8}
departure & track & own & wrong & own & wrong & own & wrong\\ \midrule
none & local & -- & 0.050 & -- & 0.060 & -- & 0.055\\
link only & local & 1.00 & 0.435 & 1.00 & 0.155 & 1.00 & 0.060\\
covariate only & local & 1.00 & 0.305 & 1.00 & 0.145 & 1.00 & 0.195\\
link + covariate & local & 0.86 & 0.000 & 1.00 & 0.000 & 1.00 & 0.000\\
\bottomrule\end{tabular}\end{table}

\begin{table}[ht]\centering\small
\caption{External validation, moderate departures followed to $n=32{,}000$: rate at which every group with a part of the departure is named (own) and at which any other group is named (wrong; nominal bound .05; Monte Carlo standard error at most 0.019). Local track: size $C\sqrt{1000/n}$; fixed track: size $C$. 400 (200 at $n=32{,}000$) data sets per cell.}\label{tab:localmo}
\begin{tabular}{llrrrrrrrr}\toprule
 & & \multicolumn{2}{c}{$n=1{,}000$} & \multicolumn{2}{c}{$n=4{,}000$} & \multicolumn{2}{c}{$n=16{,}000$} & \multicolumn{2}{c}{$n=32{,}000$}\\
\cmidrule(lr){3-4} \cmidrule(lr){5-6} \cmidrule(lr){7-8} \cmidrule(lr){9-10}
departure & track & own & wrong & own & wrong & own & wrong & own & wrong\\ \midrule
slope only & local & 0.82 & 0.150 & 0.77 & 0.087 & 0.70 & 0.085 & 0.71 & 0.055\\
link only & local & 0.90 & 0.117 & 0.90 & 0.075 & 0.85 & 0.050 & 0.82 & 0.070\\
covariate only & local & 1.00 & 0.168 & 1.00 & 0.083 & 1.00 & 0.062 & 1.00 & 0.065\\
\bottomrule\end{tabular}\end{table}

\part*{Appendix B: Simulation and application details}
\setcounter{section}{0}
\renewcommand{\thesection}{B.\arabic{section}}
\renewcommand{\thetable}{B.\arabic{table}}
\renewcommand{\thefigure}{B.\arabic{figure}}
\setcounter{figure}{0}
\setcounter{table}{0}

\section{Software}

The procedure is the R file \texttt{localize\_groups.R} of the deposit, with two functions:
\texttt{localize\_external(y, p, X, M = 999, calibration = c("montecarlo", "multiplier"), cov\_df = 5)} for frozen
predictions \texttt{p} and covariates \texttt{X}, and \texttt{localize\_insample(fit, X, B = 199, dealias = FALSE, cov\_df = 5)}
for a fitted logistic \texttt{glm}. Both return the groups named, the $p$-values of all intersection tests and of every
member. Both return an object of class \texttt{gof\_localize}, and the file \texttt{localize\_plot.R} gives it a
\texttt{plot} method: \texttt{plot(r)} draws the misfit compass and the lattice of closed tests of Figure~4 of the paper,
with the verdict and the next step written underneath (\texttt{colour = FALSE} for greys; \texttt{which = "compass"} or
\texttt{"lattice"} for one panel). Members: \textsc{intercept}, the constant; \textsc{slope}, the score projected off the constant; \textsc{link},
$(\eta^2,\eta^3)$, the two Stukel terms $\tfrac12\eta^2\mathbf 1(\eta\ge0)$ and $-\tfrac12\eta^2\mathbf 1(\eta<0)$, and
$\mathrm{ns}(\eta,4)$, each projected off $\{1,\eta\}$ (and in-sample off the model matrix); \textsc{cov}, the squares and
cubes of the standardized continuous covariates, all pairwise products of standardized covariates, and
$\mathrm{ns}(x_j,3)$ of each continuous covariate, each projected off $\{1,\eta,\mathrm{ns}(\eta,5)\}$ (and in-sample off the model
matrix). A covariate is continuous when it has more than four distinct values. Projections are weighted least
squares with weights $p(1-p)$, through a pivoted QR decomposition so that redundant columns are dropped. Member
$p$-values are truncated to $[10^{-12},1-10^{-12}]$ before the Cauchy transform (Remark~\ref{rem:clip}).

\section{Simulation design}

\paragraph{Main study (Figure 1 of the paper).} Five independent $N(0,1)$ covariates; frozen score
$\eta_0=-0.5+0.5x_1+0.8x_2+0.6x_3+0.4x_4+0.3x_5$; true $\mathrm{logit}\,\pi=\eta_0+d$ with $d=0$ (null), $C$ (intercept shift),
$(C-1)\eta_0$ (slope $C$), $C(\eta_0^2-1)$ (curved link), $C(x_1^2-1)$ (U-shape), $C\max(x_2-1,0)$ (threshold), $Cx_1x_3$
(interaction). The ``pure'' families of Tables~\ref{tab:decomp}--\ref{tab:localmo} project the same functions, on each
data set and with the weights $w_0$, off the earlier pieces: $\eta_0$ off $1$; $\eta_0^2$ off $\{1,\eta_0\}$; $x_1x_3$ off
$\{1,\eta_0,\mathrm{ns}(\eta_0,5)\}$. Seeds are fixed per cell and replicate and recorded in the deposit, with every
per-data-set result.

\paragraph{Local and moderate departures (Tables~\ref{tab:localex}--\ref{tab:localmo}).} The pure families at size
$C\sqrt{1000/n}$ (local track) or $C$ (fixed track): large sizes $C=-0.4$ (slope), $0.5$ (link), $2$ (covariate), and a
combination of link $C$ and covariate $4C$ with $C=0.5$; moderate sizes $-0.25$, $0.2$ and $0.7$ followed to
$n=32{,}000$.

\paragraph{Multiplier calibration (Table~\ref{tab:robust}).} Departures pure on the probability scale and inside
$(0,1)$ without truncation: $\pi=p_0+s\,w_0g$ with $g$ bounded on the sample and $s=0.95/\max|g|$, so
$p_0-w_0<\pi<p_0+w_0$. \emph{link\_b}: $g$ is the third column of $\mathrm{ns}(\eta_0,5)$ projected off $\{1,\eta_0\}$, so only
\textsc{link} has a part (the COV bases are projected off $\mathrm{ns}(\eta_0,5)$ and cannot read it). \emph{cov\_b}:
$g=\tanh(x_1x_3)$ projected off $\{1,\eta_0,\mathrm{ns}(\eta_0,5)\}$ and off every \textsc{link} basis, so only \textsc{cov} has a
part. The departures have fixed size, so Theorem~\ref{thm:fwer} does not apply and Theorem~\ref{thm:multiplier} does. Both
calibrations are applied to the same data sets.

\paragraph{De-aliasing (Table~\ref{tab:dealias}).} In-sample, $x_1\sim N(0,1)$, $x_2=E+0.5x_1^2$ with $E$ standard
exponential, $x_3\sim\mathrm{Bernoulli}(0.4)$; model $\eta_0=-1+0.6x_1+0.5x_2+0.5x_3$, fitted with main effects. Here
$\E(X\mid\eta_0)$ is not affine. Departures (normalized to unit $\nrm{\cdot}$): covariate structure, $0.6\,h$ with
$h=x_1x_2$ projected off $\{1,\eta_0,\mathrm{ns}(\eta_0,5)\}$; link, $0.3\,h$ with $h=\eta_0^2$ projected off $\{1,\eta_0\}$. Each
data set is analyzed with and without de-aliasing.

\paragraph{Corrupted records (Table~\ref{tab:corrupt}).} The stored data sets of a separate study of directed
tests: $x\sim U(-3,3)$, $d\sim\mathrm{Bernoulli}(0.5)$, $\mathrm{logit}\,\pi=0.6x+0.5d$, the correct model fitted, and after the outcomes
were drawn $x$ multiplied by 4 or 8 in 3 of 500 or 10 of 1,000 records. Every data set is rebuilt from its seed and
checked against the stored number of events and fitted slope. Any group named is a false one.

\section{Results of the additional studies}

\begin{table}[ht]\centering\small
\caption{Monte Carlo versus multiplier calibration, external validation, departures of fixed size that are pure on the probability scale. Own: every group with a part named; wrong: any other group named (nominal bound .05). Same data sets for both calibrations; 400 data sets per cell (200 at $n=16{,}000$).}\label{tab:robust}
\begin{tabular}{lrrrrr}\toprule & & \multicolumn{2}{c}{Monte Carlo} & \multicolumn{2}{c}{multiplier}\\
\cmidrule(lr){3-4}\cmidrule(lr){5-6} departure & $n$ & own & wrong & own & wrong\\ \midrule
none & 1{,}000 & -- & 0.045 & -- & 0.090\\
none & 4{,}000 & -- & 0.055 & -- & 0.070\\
link (probability scale) & 1{,}000 & 0.09 & 0.043 & 0.03 & 0.055\\
link (probability scale) & 4{,}000 & 0.16 & 0.033 & 0.08 & 0.037\\
link (probability scale) & 16{,}000 & 0.36 & 0.065 & 0.31 & 0.055\\
covariate (probability scale) & 1{,}000 & 0.63 & 0.045 & 0.68 & 0.080\\
covariate (probability scale) & 4{,}000 & 0.99 & 0.037 & 0.99 & 0.058\\
covariate (probability scale) & 16{,}000 & 1.00 & 0.070 & 1.00 & 0.070\\
gross miscalibration$^a$ & 1{,}000 & 0.92 & 0.003 & 0.98 & 0.120\\
gross miscalibration$^a$ & 4{,}000 & 1.00 & 0.000 & 1.00 & 0.085\\
\bottomrule\end{tabular}
\par\footnotesize{$^a$$\mathrm{logit}\,\pi=-2+0.8\eta_0$: about 11\% observed against 40\% predicted. \textsc{intercept}, \textsc{slope} and \textsc{link} have parts on the probability scale, so only \textsc{cov} counts as wrong; own: all three named.}
\end{table}

\begin{table}[ht]\centering\small
\caption{In-sample checking with a skewed and a binary covariate ($\E(X\mid\eta_0)$ not affine), $n=1{,}000$, $B=199$: rate at which each group is named, default and de-aliased link group (Proposition~\ref{prop:dealias}); 200 data sets per cell.}\label{tab:dealias}
\begin{tabular}{lrrrr}\toprule & \multicolumn{2}{c}{default} & \multicolumn{2}{c}{de-aliased}\\
\cmidrule(lr){2-3}\cmidrule(lr){4-5} departure & LINK & COV & LINK & COV\\ \midrule
none & 0.025 & 0.025 & 0.020 & 0.020\\
covariate structure & 0.055 & 0.415 & 0.030 & 0.385\\
link shape & 0.660 & 0.045 & 0.010 & 0.030\\
\bottomrule\end{tabular}\end{table}

\begin{table}[ht]\centering\small
\caption{Corrupted records, in-sample, correct model: rate at which any group, LINK or COV is named (all are false); default and de-aliased; 100 data sets per cell.}\label{tab:corrupt}
\begin{tabular}{lrrrrrr}\toprule & \multicolumn{3}{c}{default} & \multicolumn{3}{c}{de-aliased}\\
\cmidrule(lr){2-4}\cmidrule(lr){5-7} data & any & LINK & COV & any & LINK & COV\\ \midrule
$n=500$, clean & 0.07 & 0.04 & 0.03 & 0.08 & 0.05 & 0.03\\
$n=500$, $x\times4$ in 3 records & 0.46 & 0.44 & 0.34 & 0.43 & 0.41 & 0.34\\
$n=500$, $x\times8$ in 3 records & 0.49 & 0.48 & 0.43 & 0.52 & 0.50 & 0.44\\
$n=1{,}000$, clean & 0.02 & 0.01 & 0.01 & 0.06 & 0.05 & 0.01\\
$n=1{,}000$, $x\times4$ in 10 records & 0.65 & 0.63 & 0.51 & 0.64 & 0.62 & 0.50\\
\bottomrule\end{tabular}\end{table}
\begin{table}[ht]\centering\small
\caption{Burn injury mortality (1,000 patients, 40 facilities, 150 deaths). In-sample checking ($B=199$) of the main-effects model B0 and of three revisions, and cross-centre external validation over 20 random splits of the facilities into development and validation halves ($M=999$): parts named and single-group $p$-values in-sample; share of the 20 splits naming each part externally.}\label{tab:burn}
\begin{tabular}{llrrr}\toprule model & in-sample: parts named & AIC & L & C\\ \midrule
B0 main effects & \textsc{cov} & 350.5 & 0.075 & 0.005\\
B1: spline in age & \textsc{cov} & 346.6 & 0.115 & 0.010\\
B2: B1 + age $\times$ inhalation & \textsc{cov} & 338.0 & 0.440 & 0.010\\
B3: B2 + spline in burn area & none & 333.2 & 0.510 & 0.335\\
B0, de-aliased link group & \textsc{cov} & & 0.965 & 0.005\\
\midrule & \multicolumn{4}{l}{cross-centre validation: share of 20 facility splits naming}\\
model & I & S & L & C\\ \midrule
B0 main effects & 0.15 & 0.35 & 0.30 & 0.70\\
B3: B2 + spline in burn area & 0.10 & 0.35 & 0.10 & 0.75\\
\midrule & \multicolumn{4}{l}{same splits, outcomes simulated from the model fitted to all patients}\\
& \multicolumn{4}{l}{(model right by construction; mean over 25 data sets)}\\ \midrule
B0 main effects & 0.09 & 0.18 & 0.15 & 0.17\\
B3: B2 + spline in burn area & 0.12 & 0.36 & 0.29 & 0.45\\
\bottomrule\end{tabular}\end{table}
\begin{table}[ht]\centering\small
\caption{Detection at $n=1{,}000$, the same data sets as Figure~1 of the paper: rejection rate at level 0.05 of the procedure's global test (the intersection of all groups) and of single tests; under no misfit the rate is the level. A part is named in marginally fewer data sets than the global test rejects (Figure 1 of the paper). The le Cessie test was run in-sample only. Externally the Hosmer--Lemeshow test uses deciles of the frozen predictions with 10 degrees of freedom; the calibration belt is that of GiViTI. 300 data sets per cell (1,000 under no misfit).}\label{tab:compare}
\begin{tabular}{llrrrrr}\toprule setting & departure & procedure & Hosmer--Lemeshow & Spiegelhalter & belt & le Cessie\\ \midrule
external & none & 0.05 & 0.04 & 0.04 & 0.04 & --\\
external & intercept shift & 1.00 & 0.99 & 0.61 & 1.00 & --\\
external & slope 0.6 & 1.00 & 1.00 & 1.00 & 1.00 & --\\
external & curved link & 1.00 & 1.00 & 0.92 & 1.00 & --\\
external & U-shape & 1.00 & 0.58 & 0.42 & 0.65 & --\\
external & threshold & 0.96 & 0.46 & 0.18 & 0.77 & --\\
external & interaction & 1.00 & 0.56 & 0.35 & 0.53 & --\\
in-sample & none & 0.06 & 0.05 & -- & 0.06 & 0.06\\
in-sample & curved link & 1.00 & 0.96 & -- & 1.00 & 0.84\\
in-sample & U-shape & 1.00 & 0.16 & -- & 0.44 & 1.00\\
in-sample & threshold & 0.83 & 0.18 & -- & 0.38 & 0.57\\
in-sample & interaction & 1.00 & 0.28 & -- & 0.66 & 1.00\\
\bottomrule\end{tabular}\end{table}

\begin{table}[ht]\centering\small
\caption{Ways of naming the parts, on the same data sets as the main simulation: the closed procedure, naming every group whose own test has $p\le0.05$ (naive), and Bonferroni and Holm corrections over the groups' own tests. Error: rate of naming any part under no misfit; right action and own part: averages over the departure cells of the main simulation.}\label{tab:multiplicity}
\begin{tabular}{llrrrrrrrrrrrr}\toprule & & \multicolumn{4}{c}{error under no misfit} & \multicolumn{4}{c}{right action} & \multicolumn{4}{c}{own part named}\\
\cmidrule(lr){3-6}\cmidrule(lr){7-10}\cmidrule(lr){11-14} setting & $n$ & clos. & naive & Bonf. & Holm & clos. & naive & Bonf. & Holm & clos. & naive & Bonf. & Holm\\ \midrule
external & 500 & 0.037 & 0.158 & 0.032 & 0.032 & 0.742 & 0.772 & 0.774 & 0.765 & 0.806 & 0.882 & 0.813 & 0.816\\
external & 1{,}000 & 0.051 & 0.184 & 0.043 & 0.043 & 0.825 & 0.830 & 0.863 & 0.847 & 0.915 & 0.962 & 0.923 & 0.924\\
in-sample & 500 & 0.049 & 0.105 & 0.047 & 0.047 & 0.835 & 0.859 & 0.851 & 0.835 & 0.880 & 0.904 & 0.879 & 0.880\\
in-sample & 1{,}000 & 0.055 & 0.100 & 0.055 & 0.055 & 0.920 & 0.930 & 0.934 & 0.917 & 0.970 & 0.980 & 0.963 & 0.967\\
\bottomrule\end{tabular}\end{table}

\begin{table}[ht]\centering\small
\caption{Small samples, $n=200$ and $300$ (design of the main simulation): rate of naming each part and the right-action rate. 300 data sets per cell (1,000 under no misfit).}\label{tab:small}
\begin{tabular}{llrrrrrr}\toprule setting & departure & $n$ & I & S & L & C & right\\ \midrule
external & none & 200 & 0.00 & 0.01 & 0.03 & 0.01 & 0.95\\
external & none & 300 & 0.00 & 0.00 & 0.01 & 0.02 & 0.96\\
external & intercept shift & 200 & 0.36 & 0.01 & 0.04 & 0.03 & 0.33\\
external & intercept shift & 300 & 0.62 & 0.01 & 0.03 & 0.02 & 0.58\\
external & slope 0.6 & 200 & 0.08 & 0.62 & 0.18 & 0.10 & 0.42\\
external & slope 0.6 & 300 & 0.08 & 0.73 & 0.21 & 0.10 & 0.48\\
external & curved link & 200 & 0.03 & 0.36 & 0.70 & 0.12 & 0.59\\
external & curved link & 300 & 0.05 & 0.50 & 0.86 & 0.14 & 0.73\\
external & U-shape & 200 & 0.02 & 0.10 & 0.21 & 0.92 & 0.92\\
external & U-shape & 300 & 0.03 & 0.15 & 0.25 & 0.99 & 0.99\\
external & threshold & 200 & 0.02 & 0.01 & 0.02 & 0.07 & 0.07\\
external & threshold & 300 & 0.05 & 0.01 & 0.02 & 0.12 & 0.12\\
external & interaction & 200 & 0.01 & 0.09 & 0.23 & 0.68 & 0.68\\
external & interaction & 300 & 0.03 & 0.11 & 0.29 & 0.94 & 0.94\\
in-sample & none & 200 & -- & -- & 0.03 & 0.03 & 0.95\\
in-sample & none & 300 & -- & -- & 0.03 & 0.03 & 0.95\\
in-sample & curved link & 200 & -- & -- & 0.44 & 0.07 & 0.39\\
in-sample & curved link & 300 & -- & -- & 0.66 & 0.06 & 0.61\\
in-sample & U-shape & 200 & -- & -- & 0.16 & 0.94 & 0.94\\
in-sample & U-shape & 300 & -- & -- & 0.15 & 1.00 & 1.00\\
in-sample & threshold & 200 & -- & -- & 0.06 & 0.08 & 0.08\\
in-sample & threshold & 300 & -- & -- & 0.04 & 0.10 & 0.10\\
in-sample & interaction & 200 & -- & -- & 0.16 & 0.66 & 0.66\\
in-sample & interaction & 300 & -- & -- & 0.21 & 0.95 & 0.95\\
\bottomrule\end{tabular}\end{table}
\begin{table}[ht]\centering\small
\caption{The robust option (bases on normal scores) against the default. Top: corrupted records, correct model, in-sample, $B=199$: rate of naming any part (all false), 100 data sets per cell. Bottom: the departures of the main simulation at $n=1{,}000$, the same data sets: rate of naming some part and right-action rate.}\label{tab:robustopt}
\begin{tabular}{lrrrr}\toprule corrupted records & \multicolumn{2}{c}{default} & \multicolumn{2}{c}{robust}\\
\cmidrule(lr){2-3}\cmidrule(lr){4-5} & any part & \textsc{link} & any part & \textsc{link}\\ \midrule
$n=500$, clean & 0.07 & 0.04 & 0.05 & 0.02\\
$n=500$, $x\times4$ in 3 records & 0.46 & 0.44 & 0.25 & 0.12\\
$n=500$, $x\times8$ in 3 records & 0.49 & 0.48 & 0.22 & 0.12\\
$n=1{,}000$, clean & 0.02 & 0.01 & 0.03 & 0.00\\
$n=1{,}000$, $x\times4$ in 10 records & 0.65 & 0.63 & 0.42 & 0.22\\
\midrule main departures & \multicolumn{2}{c}{default} & \multicolumn{2}{c}{robust}\\
\cmidrule(lr){2-3}\cmidrule(lr){4-5} & some part & right action & some part & right action\\ \midrule
external, interaction & 1.000 & 1.000 & 1.000 & 1.000\\
external, intercept shift & 1.000 & 0.923 & 1.000 & 0.923\\
external, curved link & 1.000 & 0.883 & 1.000 & 0.867\\
external, none & 0.051 & 0.949 & 0.046 & 0.954\\
external, slope 0.6 & 1.000 & 0.603 & 1.000 & 0.613\\
external, threshold & 0.957 & 0.873 & 0.967 & 0.887\\
external, U-shape & 1.000 & 1.000 & 1.000 & 1.000\\
in-sample, interaction & 1.000 & 1.000 & 1.000 & 1.000\\
in-sample, curved link & 0.997 & 0.893 & 0.997 & 0.910\\
in-sample, none & 0.055 & 0.945 & 0.047 & 0.953\\
in-sample, threshold & 0.833 & 0.783 & 0.790 & 0.720\\
in-sample, U-shape & 1.000 & 1.000 & 1.000 & 1.000\\
\bottomrule\end{tabular}\end{table}

\section{SUPPORT}

The SUPPORT study (Knaus et al., 1995) is public as the SUPPORT2 release of the UCI Machine Learning Repository
(Harrell, 2023); it has 9,105 patients; 8,873 have complete values of
the ten predictors and the outcome (in-hospital death, 2,328 deaths). Predictors: age, mean arterial pressure, heart
rate, respiratory rate, temperature, serum creatinine, sodium, white blood cell count, number of comorbidities and
sex. M2 uses natural splines with three degrees of freedom in the seven physiological variables. Random split: a
random half (seed recorded) for development and the other half for validation. Transported: development on the 4,743
patients in the acute disease classes (ARF/MOSF, coma), validation on the 4,130 in the chronic classes (COPD/CHF/
cirrhosis, cancer). Split-first update-and-retest: the transported validation set is split at random (seed recorded) into
halves A and B of 2,065 patients; the verdict is reached on A, each update is fitted on A, and the updated predictions
are tested on B. An earlier run on a different split, which chose the rung from all 4,130 chronic patients and so used
the test half twice, named \textsc{cov} after five of the six recalibrations (not after the flexible recalibration of M2,
$p=0.048$); it is in the deposit. Neural
network: \texttt{nnet} with five hidden units, weight decay 0.01, entropy fitting, standardized inputs. Timing: on one core, the in-sample procedure for M0 on all 8,873 patients ($B=199$) took 67 seconds and the external
validation of the transported M0 ($n=4{,}130$, $M=999$) 1 second. The projection test of Liu et al.\ (2024), run on the same
8,873 patients in a separate study of ours, took 9 hours 53 minutes. The planning quantities of Corollary~\ref{cor:power} for each external analysis are in Table~\ref{tab:plan}. Web
Figure~\ref{fig:calib} shows the calibration curves of the transported models on half B after each update.

\begin{figure}[ht]
\centering\includegraphics[width=\textwidth]{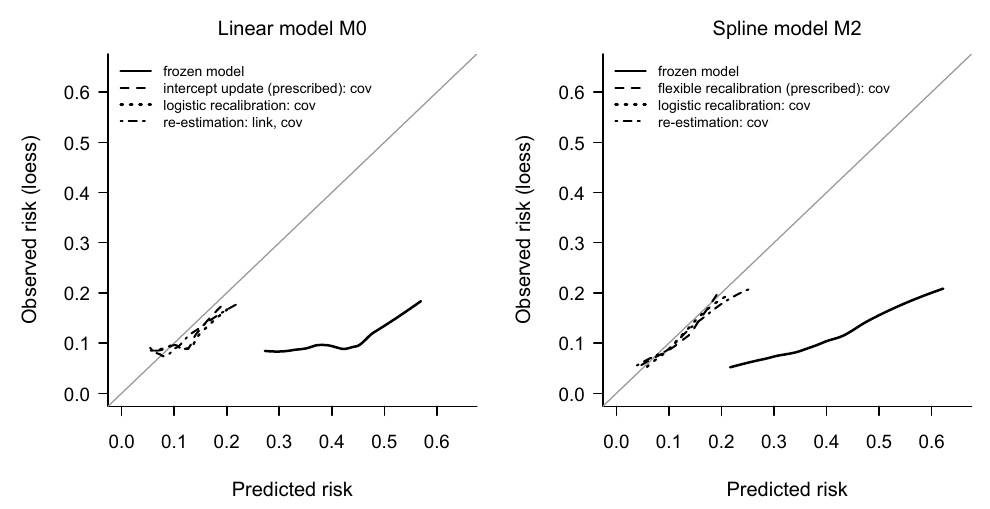}
\caption{Transported SUPPORT models on half B of the chronic patients: calibration curves (loess of the outcome on the
predicted risk) of the frozen model and of each update fitted on half A, with the parts the procedure named on half B.
Grey line: perfect calibration. After every update the curve lies close to the diagonal, yet \textsc{cov} is named.}
\label{fig:calib}
\end{figure}

\begin{table}[ht]\centering\small
\caption{SUPPORT, external analyses: smallest mean miscalibration $\delta$ and logit-scale slope error $\epsilon$ that the single-member tests detect with power .80 at level .05 (Corollary~\ref{cor:power}), computed from the covariates and frozen predictions only, with the observed and mean predicted risk.}\label{tab:plan}
\begin{tabular}{llrrrrr}\toprule setting & model & $n$ & $\delta$ & $\epsilon$ & observed & predicted\\ \midrule
external random & M0 & 4437 & 0.018 & 0.214 & 0.260 & 0.266\\
external random & M2 & 4437 & 0.018 & 0.161 & 0.260 & 0.266\\
external transported & M0 & 4130 & 0.021 & 0.306 & 0.104 & 0.400\\
external transported & M2 & 4130 & 0.021 & 0.226 & 0.104 & 0.377\\
external random & neural network & 4437 & 0.018 & 0.142 & 0.260 & 0.266\\
external transported & neural network & 4130 & 0.021 & 0.197 & 0.104 & 0.376\\
\bottomrule\end{tabular}\end{table}

\section*{References for the appendices}
{\small\begin{itemize}[leftmargin=1.2em,itemsep=1pt,label={}]
\item de Boor, C. (2001). \emph{A Practical Guide to Splines}, revised edition. New York: Springer.
\item Dufour, J.-M. (2006). Monte Carlo tests with nuisance parameters: a general approach to finite-sample inference and nonstandard asymptotics. \emph{Journal of Econometrics} \textbf{133}, 443--477.
\item Escanciano, J. C. (2006). A consistent diagnostic test for regression models using projections. \emph{Econometric Theory} \textbf{22}, 1030--1051.
\item Harrell, F. E., Jr. (2023). SUPPORT2 [Dataset]. UCI Machine Learning Repository, https://archive.ics.uci.edu/dataset/880/support2.
\item Hosmer, D. W., Hosmer, T., le Cessie, S., and Lemeshow, S. (1997). A comparison of goodness-of-fit tests for the logistic regression model. \emph{Statistics in Medicine} \textbf{16}, 965--980.
\item Janssen, K. J. M., Moons, K. G. M., Kalkman, C. J., Grobbee, D. E., and Vergouwe, Y. (2008). Updating methods improved the performance of a clinical prediction model in new patients. \emph{Journal of Clinical Epidemiology} \textbf{61}, 76--86.
\item Knaus, W. A., Harrell, F. E., Jr., Lynn, J., Goldman, L., Phillips, R. S., Connors, A. F., Jr., et al. (1995). The SUPPORT prognostic model: objective estimates of survival for seriously ill hospitalized adults. \emph{Annals of Internal Medicine} \textbf{122}, 191--203.
\item Lehmann, E. L. and Romano, J. P. (2005). \emph{Testing Statistical Hypotheses}, 3rd edition. New York: Springer.
\item Liu, H., Li, X., Chen, F., H\"ardle, W., and Liang, H. (2024). A comprehensive comparison of goodness-of-fit tests for logistic regression models. \emph{Statistics and Computing} \textbf{34}, 175.
\item Liu, Y. and Xie, J. (2020). Cauchy combination test: a powerful test with analytic $p$-value calculation under arbitrary dependency structures. \emph{Journal of the American Statistical Association} \textbf{115}, 393--402.
\item Mammen, E. (1993). Bootstrap and wild bootstrap for high dimensional linear models. \emph{The Annals of Statistics} \textbf{21}, 255--285.
\item Marcus, R., Peritz, E., and Gabriel, K. R. (1976). On closed testing procedures with special reference to ordered analysis of variance. \emph{Biometrika} \textbf{63}, 655--660.
\item Pierce, D. A. (1982). The asymptotic effect of substituting estimators for parameters in certain types of statistics. \emph{The Annals of Statistics} \textbf{10}, 475--478.
\item Riley, R. D., Debray, T. P. A., Collins, G. S., Archer, L., Ensor, J., van Smeden, M., and Snell, K. I. E. (2021). Minimum sample size for external validation of a clinical prediction model with a binary outcome. \emph{Statistics in Medicine} \textbf{40}, 4230--4251.
\item Steyerberg, E. W., Borsboom, G. J. J. M., van Houwelingen, H. C., Eijkemans, M. J. C., and Habbema, J. D. F. (2004). Validation and updating of predictive logistic regression models: a study on sample size and shrinkage. \emph{Statistics in Medicine} \textbf{23}, 2567--2586.
\item Pregibon, D. (1980). Goodness of link tests for generalized linear models. \emph{Applied Statistics} \textbf{29}, 15--24.
\item Stukel, T. A. (1988). Generalized logistic models. \emph{Journal of the American Statistical Association} \textbf{83}, 426--431.
\item Stute, W., Gonz\'alez Manteiga, W., and Presedo Quindimil, M. (1998). Bootstrap approximations in model checks for regression. \emph{Journal of the American Statistical Association} \textbf{93}, 141--149.
\item Stute, W. and Zhu, L.-X. (2002). Model checks for generalized linear models. \emph{Scandinavian Journal of Statistics} \textbf{29}, 535--545.
\item Sur, P. and Cand\`es, E. J. (2019). A modern maximum-likelihood theory for high-dimensional logistic regression. \emph{Proceedings of the National Academy of Sciences} \textbf{116}, 14516--14525.
\item Van Calster, B., Nieboer, D., Vergouwe, Y., De Cock, B., Pencina, M. J., and Steyerberg, E. W. (2016). A calibration hierarchy for risk models was defined: from utopia to empirical data. \emph{Journal of Clinical Epidemiology} \textbf{74}, 167--176.
\item van der Vaart, A. W. (1998). \emph{Asymptotic Statistics}. Cambridge: Cambridge University Press.
\item van der Vaart, A. W. and Wellner, J. A. (1996). \emph{Weak Convergence and Empirical Processes}. New York: Springer.
\item Wald, A. (1943). Tests of statistical hypotheses concerning several parameters when the number of observations is large. \emph{Transactions of the American Mathematical Society} \textbf{54}, 426--482.
\item White, H. (1982). Maximum likelihood estimation of misspecified models. \emph{Econometrica} \textbf{50}, 1--25.
\item Wu, C. F. J. (1986). Jackknife, bootstrap and other resampling methods in regression analysis. \emph{The Annals of Statistics} \textbf{14}, 1261--1295.
\end{itemize}}

\end{document}